\documentclass[sigconf,screen,nonacm]{acmart}
\AtBeginDocument{%
  }

\acmYear{2026}\copyrightyear{2026}
\setcctype[4.0]{by}
\acmConference[SPAA '26]{38th ACM Symposium on Parallelism in Algorithms and Architectures}{July 6--10, 2026}{London, United Kingdom}
\acmBooktitle{38th ACM Symposium on Parallelism in Algorithms and Architectures (SPAA '26), July 6--10, 2026, London, United Kingdom}
\acmDOI{10.1145/3816782.3819218}
\acmISBN{979-8-4007-2761-0/26/07}
\usepackage{mathtools}
\usepackage{algorithm}
\usepackage{algpseudocode}
\algblockdefx{MRepeat}{EndRepeat}{\textbf{repeat}}{}
\algnotext{EndRepeat}
\usepackage{nicefrac}
\usepackage{tabularx}
\usepackage{makecell}
\usepackage{subcaption}
\usepackage{xspace}
\usepackage[capitalize]{cleveref}
\usepackage{standalone}
\usepackage{tikz}
\usetikzlibrary{calc}
\usetikzlibrary{decorations.pathreplacing}

\usepackage{thmtools, thm-restate} % restatable

\definecolor{blue}{HTML}{648FFF}
\definecolor{red}{HTML}{DC267F}
\definecolor{yellow}{HTML}{FFB000}
\definecolor{gapcolor}{HTML}{800000}

\newsavebox{\bigimage}
\newsavebox{\bigimageTwo}

\tikzset{font=\fontsize{22}{0}\selectfont\bfseries\boldmath}

\newcommand{\case}[1]{\noindent\textbf{Case #1:}}

\newenvironment{subproof}[1][\proofname]{%
  \begin{proof}[#1]%
    }{%
  \end{proof}%
}

\AtEndPreamble{%
  \theoremstyle{acmplain}
  \newtheorem{observation}[theorem]{Observation}
  \newtheorem{remark}[theorem]{Remark}
  \newtheorem{claim}[theorem]{Claim}
  \theoremstyle{acmdefinition}
  \newtheorem{properties}[theorem]{Properties}
}

\Crefname{observation}{Observation}{Observations}
\Crefname{remark}{Remark}{Remarks}
\Crefname{claim}{Claim}{Claims}
\Crefname{properties}{Properties}{Properties}

\newcommand{\opt}{\textnormal{\texttt{OPT}}}
\newcommand{\LS}{\textsc{GreedyScheduling}\xspace}
\newcommand{\PTS}{\textsc{Parallel Task Scheduling}\xspace}
\newcommand{\MCS}{\textsc{Multiple Cluster Scheduling}\xspace}

\newcommand{\calJ}{\mathcal{J}}

\newcommand{\pmax}{p_{\max}}

\newcommand{\calS}{\mathcal{S}}
\newcommand{\calM}{\mathcal{M}}
\newcommand{\calB}{\mathcal{B}}
\newcommand{\calU}{\mathcal{U}}
\newcommand{\calT}{\mathcal{T}}
\newcommand{\calR}{\mathcal{R}}
\newcommand{\deficit}{\text{deficit}}

\newcommand{\stackrectright}[2]{%
    \path (\xright, \ycur) coordinate (start)
    ++(-#1, 0) coordinate (leftbase)
    ++(0, #2) coordinate (lefttop)
    ++(#1, 0) coordinate (righttop);
    \filldraw (start) -- (leftbase) -- (lefttop) -- (righttop) -- cycle;
    \pgfmathsetmacro{\ycur}{\ycur + #2}
}

\newcommand{\stackrectleft}[3]{% {width} {height} {direction}
    \ifthenelse{\equal{#3}{up}}{%
        \path (\xleft, \ycur) coordinate (start)
        ++(#1, 0) coordinate (rightbase)
        ++(0, #2) coordinate (righttop)
        ++(-#1, 0) coordinate (lefttop);
        \filldraw (start) -- (rightbase) -- (righttop) -- (lefttop) -- cycle;
        \pgfmathsetmacro{\ycur}{\ycur + #2}
    }{%
        \path (\xleft, \ycur) coordinate (start)
        ++(#1, 0) coordinate (rightbase)
        ++(0, -#2) coordinate (rightbottom)
        ++(-#1, 0) coordinate (leftbottom);
        \filldraw (start) -- (rightbase) -- (rightbottom) -- (leftbottom) -- cycle;
        \pgfmathsetmacro{\ycur}{\ycur - #2}
    }
}

\def\jobSA{{3.1}{0.5}}
\def\jobSB{{3.1}{0.5}}
\def\jobSC{{3.0}{0.5}}
\def\jobSD{{2.8}{0.5}}
\def\jobSE{{2.7}{0.5}}
\def\jobSF{{2.5}{0.5}}
\def\jobMA{{4.5}{1}}
\def\jobMB{{4.5}{0.5}}
\def\jobMC{{4.1}{0.6}}
\def\jobMD{{4}{0.3}}
\def\jobME{{3.5}{1}}
\def\jobMF{{3.4}{0.5}}

\def\jobBA{{7.4}{1}}
\def\jobBB{{7}{0.5}}
\def\jobBC{{6.5}{1}}
\def\jobBD{{5.9}{0.7}}
\def\jobBE{{5.8}{1}}
\def\jobMG{{3.9}{1}}
\def\jobSG{{2.7}{0.5}}

\begin{document}

\title{Improved Approximation Algorithms for Parallel Task Scheduling and Multiple Cluster Scheduling}
% typeset the header of the contribution

\author{Bennet Edler}
\orcid{0009-0002-5067-6781}
\email{stu235896@mail.uni-kiel.de}
\affiliation{%
    \institution{Kiel University}
    \city{Kiel}
    \country{Germany}
}
\author{Klaus Jansen}
\orcid{0000-0001-8358-6796}
\email{kj@informatik.uni-kiel.de}
\affiliation{%
    \institution{Kiel University}
    \city{Kiel}
    \country{Germany}
}
\author{Felix Ohnesorge}
\orcid{0009-0003-8023-3380}
\email{foh@informatik.uni-kiel.de}
\affiliation{%
    \institution{Kiel University}
    \city{Kiel}
    \country{Germany}
}
\author{Lis Pirotton}
\orcid{0009-0001-4984-3696}
\email{lpi@informatik.uni-kiel.de}
\affiliation{%
    \institution{Kiel University}
    \city{Kiel}
    \country{Germany}
}

% \affiliation{%
% \institution{The Th{\o}rv{\"a}ld Group}
% \city{Hekla}
% \country{Iceland}}


\begin{abstract}
    In the problem of \PTS (PTS), we are asked to schedule \(n\) jobs, each with a fixed processing time and machine requirement,
    such that the completion time of the last job is minimized.
    Jansen and Rau (2019) presented an algorithm for PTS that achieves an approximation ratio of \((3/2)\opt + p_{\max}\).
    They additionally posed the open question whether an approximation ratio of \((4/3)\opt + \pmax\) is possible.
    In this work, we present such an algorithm with a running time of \(O(n\log n)\).

    The problem of \MCS (MCS) is a natural extension of PTS where we are given \(N\) clusters each of \(m\) machines to schedule jobs.
    Jansen and Rau (2019) adapted their PTS algorithm to MCS with the following results:
    (1) a 2 approximation, and (2) a near-linear 9/4 approximation if \(N\) is divisible by 3.
    We improve the running time of their 2-approximation and generalize the 9/4 approximation to the general case.
    The 2-approximation for MCS is tight, since one cannot hope for an approximation ratio better than 2, unless P=NP [Zhuk, 2006].
    In addition to our theoretical results, we implement our algorithm and show its practical applicability.
    %\keywords{Parallel Scheduling \and Approximation Algorithms \and Cluster Scheduling.}
\end{abstract}

% TODO: Gernerate Classifications

% \ccsdesc[300]{Theory of computation~Streaming, sublinear and near linear time algorithms}
%     \includegraphics[width=\textwidth]{sampleteaser}
%     \caption{Seattle Mariners at Spring Training, 2010.}
%     \Description{Enjoying the baseball game from the third-base
%         seats. Ichiro Suzuki preparing to bat.}
%     \label{fig:teaser}


\maketitle

\section{Introduction}
In this paper, we consider the problem of \PTS (PTS) where one is given a set $\calJ = \{1, \dots, n\}$ of jobs and \(m\) identical machines.
Each job $j \in \calJ$ is specified by a processing time $p_j \in \mathbb{Q}_+$ and a machine requirement $q_j \in \mathbb{N}_{\le m}$.
All jobs shall be executed non-preemptively and not necessarily contiguously on $m$ identical machines.
A solution (schedule) is a function $\sigma:\calJ \mapsto \mathbb{Q}_+$, assigning starting times to each job.
Feasible schedules are required to satisfy the resource constraint, that is,
at any time $\tau \in \mathbb{Q}_+$ the total number of machines allocated to concurrently running jobs does not exceed $m$:
\begin{align*}
    \sum_{\substack{j\in \calJ \\ \sigma(j)\le \tau < \sigma(j)+p_j}} q_j \le m.
\end{align*}
The objective is to find a schedule with minimal height (makespan), defined as $h(\sigma) = \max_{j\in \calJ} (\sigma(j) + p_j)$.

The PTS problem arises in a variety of applications, including batch processing, multiprocessor job scheduling
and resource allocation in high-performance and cloud computing systems.
From a theoretic point of view, the problem is known to be strongly $NP$-hard,
even when the number of machines is a small constant ($m \ge 4$)~\cite{DBLP:journals/siamdm/DuL89a,HenningJRS17}.
As exact polynomial-time algorithms are therefore unlikely to exist, the focus of research has shifted toward the development of approximation algorithms.
Those approaches aim to compute feasible schedules with bounded makespan relative to the optimum.
More precisely, an algorithm is said to be an \(\alpha\)-approximation algorithm with $\alpha \geq 1$ if for every instance $I$, the height $h(A(I))$ of the resulting schedule satisfies $h(A(I)) \le \alpha \opt(I)$, where $\opt(I)$ denotes the optimal makespan for instance $I$.
The approximation ratio may consist of both a multiplicative $\alpha \geq 1$ and an additive $\beta \geq 0$ term.
Then the height $A(I)$ satisfies $A(I) \le \alpha \opt(I)+\beta$ for all instances $I$.

A natural generalization of PTS is the \MCS (MCS) problem.
Here, the computational resources are distributed across $N \ge 1$ disjoint clusters, each consisting of $m$ identical machines.
The input remains a set of parallel jobs, but the scheduler must now determine both the assignment of jobs to clusters and a feasible schedule within each cluster.
The goal is to minimize the global makespan, i.e., the maximum completion time of any job across all clusters.
The classical PTS problem corresponds to the special case $N=1$, implying that MCS inherits the computational hardness of PTS.
The MCS model captures many computing environments where resources are physically or logically distributed, such as cloud architectures.

\subsection{Related work}
The problem of \PTS, commonly denoted as \(P|\text{size}_j|C_{\max}\) for arbitrarily many machines
and \(Pm|\text{size}_j|C_{\max}\) when considering a constant number of machines, has been studied extensively due to its theoretical complexity
and practical relevance.
The structure and complexity of PTS problems were studied in foundational work by Du and Leung~\cite{DBLP:journals/siamdm/DuL89a}.
For \(Pm|\text{size}_j|C_{\max}\), they proved strong NP-hardness when $m\ge 5$ while the problem can be solved in pseudo-polynomial time when $m\le 3$.
The case $m=4$ remained open for many years until Henning, Jansen, Rau and Schmarje~\cite{HenningJRS17} proved it to be strongly NP-complete in 2017.

A polynomial-time approximation scheme (PTAS) is an algorithm that, for any fixed $\varepsilon > 0$, computes a solution that is within a factor of $(1+\varepsilon)$ of the optimal solution, and runs in polynomial time regarding the input size.
Jansen and Porkolab~\cite{JansenP99} proposed a PTAS for PTS with a constant number of machines \(m\).
In 2008, Jansen and Thöle~\cite{JansenT08} gave a PTAS where the number of machines $m$ is polynomial in the number of jobs $n$,
i.e., $m \le \text{poly}(n)$.
Recently, Jansen and Rau~\cite{JansenR19} presented an AEPTAS with running time $O(n)\cdot f(\frac{1}{\varepsilon})$
and approximation ratio $(1+\varepsilon)\opt+\pmax$; however, the function \(f\) is triple-exponential, a point we will argue in more detail later. %, here the \(O_\varepsilon(1)\) hides large constants, which make the algorithm impractical.\todo{das umformulieren.}
If $m$ is arbitrarily large, there is no approximation algorithm with an approximation ratio \((1.5 - \delta)\opt\) for any \(\delta > 0\)~\cite{drozdowski1995no32approx,Johannes06}, also no asymptotic approximation algorithm exists with performance guarantee \(A(I) \leq \alpha \opt(I) + \beta \), for \(\alpha < 1.5\) and \(\beta \) being a polynomial in \(n\)~\cite{Johannes06}.

A list-based 2-approximation for the resource-constrained scheduling problem, which generalizes PTS,
was first proposed by Garey and Graham~\cite{DBLP:journals/siamcomp/GareyG75}.
This framework was applied to PTS by Ludwig and
Tiwari~\cite{DBLP:conf/soda/LudwigT94} and Turek, Wolf, and Yu~\cite{DBLP:conf/spaa/TurekWY92}.
In 2012, Jansen~\cite{Jansen12} presented an algorithm with approximation ratio $(\frac{3}{2}+\varepsilon)\opt$ and running time
$O(n\log n) + f(\frac{1}{\varepsilon})$.
A particular focus is on the work of Jansen and Rau~\cite{JansenR19}, who developed a linear-time algorithm with approximation ratio
$\frac{3}{2}\opt+\pmax$, and additionally posed the open question of whether a fast $\frac{4}{3}\opt+\pmax$-approximation algorithm exists
-- a question we answer in this work.

PTS is a generalization of the well-known \textsc{Bin Packing} Problem where one is given a set of items with sizes at most 1
and the goal is to pack all items into a minimum number of bins of capacity 1.
If all processing times in PTS are equal to 1, the problem reduces to \textsc{Bin Packing}.
Over the years, numerous approximation algorithms have been developed for \textsc{Bin Packing}
(see e.g.,~\cite{VegaL81,KarmarkarK82,MARTEL1985189,BekesiGK00,BerghammerR03,coffman2013bin}).

For the \MCS problem, there does not exist an algorithm with approximation ratio better than 2, unless $P = NP$,
as proven by Zhuk~\cite{zhuk2006approximate}.
In response to this hardness result, Ye, Han and Zhang~\cite{YeHZ11} developed a $(2+\varepsilon)$-approximation algorithm.
This was further improved by Jansen and Trystram~\cite{DBLP:journals/endm/JansenT16}, who give a 2-approximation with a worst case running time of \(\Omega(n^{256})\).

\begin{table}[htbp]
    \centering
    \setlength{\tabcolsep}{7pt}

    \begin{tabular}{lll}
        \toprule
        \textbf{Ratio}                   & \textbf{Runtime}                                   & \textbf{Source}                         \\
        \midrule
        \(5/2\)                  & \(O(n(N + \log(n)) \log(np_{\max}))\) & \cite{DBLP:conf/ifipTCS/BougeretDJOT10} \\
        \(7/3\)                  & \(O(\log{(np_{\max})}N(n + \log{(n)}))\)           & \cite{DBLP:journals/tcs/BougeretDTJR15} \\
        \(2\)                            & \(\Omega(n^{256})\)                                & \cite{DBLP:journals/endm/JansenT16}     \\
        \(2\)                            & \(O(n\log n) + \Omega(1.9 \cdot 10^{2184})\)             & \cite{JansenR19}                        \\
        \(9/4 + O(1/N)\) & \(O(n\log n)\)                                     & \cite{JansenR19}                        \\
        \(2 + O(1/N)\)           & \(O(n\log n)\)                                     & This work                               \\
        \(9/4\)                  & \(O(n\log n)\)                                     & This work                               \\
        \bottomrule
    \end{tabular}
    \caption{Algorithms for MCS. Asymptotic ratios are denoted using \(O(1/N)\), indicating convergence to the dominant constant as the number of clusters \(N\) grows.}
    \label{table:related_work}
\end{table}

As the running time of the previously mentioned approaches is impractically large, research has focused on reducing the running time of approximation algorithms for MCS,
often at the cost of slightly worse approximation ratios.
For example, Bougeret, Dutot, Jansen, Otte and Trystram~\cite{DBLP:conf/ifipTCS/BougeretDJOT10} proposed a 2.5-approximation algorithm with improved efficiency,
achieving a running time of \(O(n(N + \log(n)) \log(np_{\max}))\).
Building on this, Bougeret, Dutot, Jansen, Robenek and Trystram~\cite{DBLP:journals/tcs/BougeretDTJR15} introduced a $\frac{7}{3}$-\allowbreak ap\-proxi\-ma\-tion algorithm with running time \(O(\log{(np_{\max})}N(n + \log{(n)}))\).
They also provided a fast algorithm with an approximation ratio of \(2\) for instances where all jobs require fewer than \(m/2\) machines.
More recently, Jansen and Rau~\cite{JansenR19} presented an $O(n)$-algorithm achieving a 2-approximation under the condition that the number of clusters $N$
is strictly larger than 2. If $N\in\{1,2\}$, the running time increases to $O(n\log n)$.
While theoretically appealing, the algorithm uses their AEPTAS.
Therefore, the \(O\)-notation again hides large constants that limit practical applicability.
Motivated by this, they additionally provide a 9/4 approximation which is asymptotic without these large constants.
Note that by asymptotic we mean that the ratio holds for $N\to \infty$.
For the case of $N\equiv0\mod3$, this algorithm achieves an absolute ratio of 9/4.

\subsection{Our contribution}
In this work, we focus on practically efficient algorithms that avoid large hidden constants and deliver fast execution times in practice.
Our first contribution is a new approximation algorithm for PTS.
\begin{restatable}{theorem}{pts}
    \label{thm:pts}
    There exists a practically efficient algorithm for PTS with an approximation ratio of \(\frac{4}{3}\opt + \pmax \) and runtime \(O(n \log n)\).
\end{restatable}
This result advances the state-of-the-art in multiple ways.
Compared to the fast \(\frac{3}{2}\opt + \pmax\) algorithm of Jansen and Rau~\cite{JansenR19},
we provide an improved approximation guarantee and resolve their open question regarding the existence of such an algorithm.
Moreover, in comparison to the AEPTAS in~\cite{JansenR19}, with approximation ratio of \((1 + \varepsilon)\opt +\pmax\) we reduce the running time significantly.
In particular, their running time of \(O(n \log(\frac{1}{\varepsilon})) + 1/\varepsilon^{1/\varepsilon^{O(1/\varepsilon)}}\) resolves to \(O(n) + \Omega(7.5 \cdot 10^{12})\), for \(\varepsilon = \frac{1}{3}\).
Additionally, we extend our result for PTS to MCS using techniques similar to~\cite{JansenR19} and provide an algorithm with an asymptotic approximation ratio.
By that we mean that we approach the ratio when $N$ increases.
\begin{restatable}{theorem}{mcsAsymp}
    \sloppy
    \label{thm:mcs}
    There exists a practically efficient algorithm for MCS with an asymptotic approximation ratio of \(2\) and runtime \(O(n\log n)\).
\end{restatable}
The existing asymptotic 2 approximation in \cite{JansenR19} again hides large constants in the running time.
In particular, the O-notation in \cite{JansenR19} hides constants of size at least \({5^5}^5 \approx 1.9 \cdot 10^{2184}\).
In comparison, our algorithm is therefore more practical.
Note that an approximation ratio of 2 is best possible~\cite{zhuk2006approximate}.
However, similar to~\cite{JansenR19}, the approximation ratio might be as large as 5 when \(N\) is small.
This motivates the following result, which gives an absolute approximation ratio for MCS.
\begin{restatable}{theorem}{mcs}
    \label{thm:mcs_9_4}
    There exists a practically efficient algorithm for MCS with an approximation ratio of $\frac{9}{4}$ and runtime \(O(n\log n)\).
\end{restatable}
Note that~\cite{JansenR19} provide an MCS algorithm with asymptotic approximation ratio of \(\frac{9}{4}\); however their algorithm has an absolute worst-case ratio of \(2.5\), if \(N=2\).
In comparison, we get an asymptotic ratio of \(2\), and a worst case ratio of \(\frac{9}{4}\).

\section{Notation and Preliminaries}
All job widths can be scaled, such that they lie in the interval $(0,1]$ by dividing by $m$.
Thus, instead of having $m$ machines and integer machine requirements,
we aim to fill a strip of width 1 with items of width $q_i \in (0,1],$ for all $i \in \{1, \dots, n\}$.
% Due to the context of the problem, we continue to refer to machine usage at a given time, by which we mean the number of (now fractional) machines being used.
This is actually a generalization of PTS, as we do not require the job width to be a multiple of \(\frac{1}{m}\).% \todo{ref zu strip packing?}

\subsection{Notation}\label{sec:notation}
We refer to $q_j$, the machine requirement of job $j$, as the job's width. Intuitively, if job \(j\) is thinner [wider] than job \(k\),
then \(q_j < q_k\) [\(q_j > q_k\)]. Similarly, we refer to the processing time of a job as its height.

During this work, we partition the set of jobs \(\calJ\) into the following subsets, based on their machine requirements.
Note that the indicated colors are used in the subsequent figures.
\begin{itemize}
    \item \makebox[3cm]{Tiny jobs (gray):\hfill}\(\calT = \{j\in\calJ : q_j\in (0,\nicefrac{1}{4}]\}\)
    \item \makebox[3.04cm]{Small jobs (blue):\hfill}\(\calS = \{j\in\calJ : q_j\in (\nicefrac{1}{4},\nicefrac{1}{3}]\}\)
    \item \makebox[2.91cm]{Medium jobs (pink):\hfill}\(\calM = \{j\in\calJ : q_j\in (\nicefrac{1}{3},\nicefrac{1}{2}]\}\)
    \item \makebox[3.02cm]{Big jobs (yellow):\hfill}\(\calB = \{j\in\calJ : q_j\in (\nicefrac{1}{2},1]\}\)
\end{itemize}

\begin{definition}[schedule]
    A (partial) schedule $\sigma \colon \mathcal{J}' \to \mathbb{Q}_+$ assigns starting times to a subset of jobs $dom(\sigma)\coloneqq \mathcal{J}' \subseteq \mathcal{J}$.
    \begin{itemize}
        \item $h(\sigma)\coloneq\max_{j\in \mathcal{J}'} \sigma(j)+p_j$ denotes the \emph{height} of $\sigma$.
        \item $\opt(\mathcal{J}')$ denotes the minimal/optimal height. We use $\opt$ if $\mathcal{J}'=\mathcal{J}$.
        \item \(q(\sigma, \tau) \coloneq \sum_{j\in \mathcal{J}', \tau \in [\sigma(j), \sigma(j)+p_j )}q_j\) denotes the \emph{machine utilization} at time $\tau$.
        \item \(\deficit(\sigma) \coloneqq \int_{0}^{h(\sigma)} \max(0,\frac{3}{4} - q(\sigma, \tau)) d\tau\) is the \emph{area deficit}.
    \end{itemize}
\end{definition}

The area deficit is defined such that \(\deficit(\sigma) = 0\) implies \(h(\sigma) \leq \frac{4}{3}\opt\) and $\deficit(\sigma) \leq \frac{3}{4}\pmax$ implies $h(\sigma) \leq \frac43\opt+\pmax$.
\begin{observation}
    \label{obs:missing_area_to_ratio}
    Any schedule \(\sigma\) with an area deficit of at most $\frac{3}{4}\pmax$ has: $ h(\sigma) \leq \frac{4}{3}\opt+p_{\max}$.
\end{observation}
\begin{proof}
    Consider an optimal schedule $\sigma^*$.
    Any schedule \(\sigma\), scheduling the same set of jobs as \(\sigma^*\), has $\text{area}(\sigma)=\text{area}(\sigma^*)$.
    Further, $\text{area}(\sigma^*)\leq 1\cdot h(\sigma^*)$.

    Since the integral of the machine utilization over time equals the total area ($\int_{0}^{h(\sigma)} q(\sigma, \tau) d\tau = \text{area}(\sigma)$), this yields:
    $$\deficit(\sigma) \geq \frac{3}{4}h(\sigma) - \text{area}(\sigma)$$

    By assumption, the area deficit of $\sigma$ is at most $\frac{3}{4}\pmax$.
    This implies $\text{area}(\sigma)\geq \frac34 h(\sigma)-\frac34 \pmax$.
    Combining all three statements implies $h(\sigma^*)\geq \text{area}(\sigma)\geq \frac34h(\sigma)-\frac34 \pmax$, which is equivalent to $h(\sigma) \leq \frac43 h(\sigma^*)+\pmax$.
\end{proof}

To aid the density arguments in our analysis, we define the \emph{sparse intervals} of a schedule \(\sigma \) as maximal time-intervals where the machine usage is less than \(\frac{3}{4}\) at any time.

\begin{definition}[Sparse Intervals]
    A sparse interval $I_k$ is a time-interval $[a, b)$ of maximal size during which the machine usage is less than $\frac34$ at any time.
    \begin{itemize}
        \item \(|I_k| = b - a\) denotes the \emph{height} of \(I_k\).
    \end{itemize}
\end{definition}

In a schedule \(\sigma\), a sparse interval \(I_k\) is considered \emph{accessible} if it is positioned at either the bottom (starting at \(0\)) or the top (ending at height \(h(\sigma)\)) of the schedule.
This property is particularly useful when combining multiple partial schedules.
Furthermore, we define an accessible interval as \emph{accessible monotone} if its machine usage follows a specific trend: it must be monotonically decreasing if it ends at \(h(\sigma)\), or monotonically increasing if it starts at \(0\).
The machine usage in a sparse interval is defined as $q(I_k)\coloneq \bigcup_{\tau\in I_k}q(\sigma,\tau)$.
Similarly, we use \(q_{\min}(I_k) \coloneqq \min(q(I_k))\) and \(q_{\max}(I_k) \coloneqq \max(q(I_k))\) to denote the minimum and maximum machine utilization in an interval \(I_k\), respectively.

\subsection{Preliminaries}
In many of our algorithms, we utilize a simple adaptation of the well-known \emph{List Scheduling} algorithm~\cite{DBLP:journals/siamcomp/GareyG75}, which we refer to as \LS.

\begin{algorithm}
    \caption{\LS}
    \label{alg:greedy}
    \begin{algorithmic}[1]
        \State Let $\mathcal{U}$ be the set of all unscheduled jobs.
        \While{$\mathcal{U} \neq \emptyset$}
        \State Find the earliest time $\tau$ where a job in $\mathcal{U}$ can be scheduled.
        \State Let $F(\tau) \subseteq U$ be the jobs that can be scheduled at time $\tau$.
        \State Schedule a job $j^* \in F(\tau)$ with the maximum width at \(\tau\).
        \State $\mathcal{U} \gets \mathcal{U} \setminus \{j^*\}$

        \EndWhile
    \end{algorithmic}
\end{algorithm}

During this work, we use two kinds of \LS.

\paragraph{\textbf{\LS in $[a,b)$}}
In this variant, the algorithm is restricted to a specific time interval defined by \(a\) and \(b\).
We apply \LS as usual, but a job is only scheduled if it is processed entirely in the interval $[a,b)$.

                \paragraph{\textbf{\LS on $k$ stacks}}
                We use this variant only when the jobs we want to place have widths in $(\frac{1}{k+1},\frac{1}{k}]$.
Because exactly \(k\) such jobs can be processed simultaneously, \LS naturally groups them into a schedule of \(k\) parallel stacks.
The machine utilization is monotonically decreasing from the bottom to the top of each stack.

\begin{observation}
    Both variants of \LS can be implemented in $O(n\log(n))$.
\end{observation}

\section{Algorithm for \PTS}\label{sec:pts}

In this section, we present an algorithm that achieves an approximation ratio of \(\frac{4}{3}\opt + p_{\max}\) for PTS.

\begin{figure}[ht]
    \centering

    \setcounter{subfigure}{2}

    \sbox{\bigimage}{%
        \begin{subfigure}[b]{0.48\linewidth}
            \centering
            \begin{tikzpicture}

    %outline of the schedule
    \def\w{10}
    \def\h{6.2}

    %Sigma 1T
    \begin{scope}[fill opacity=0.7]

        \filldraw[fill=red, ultra thick] (0, 0) -- (0.45*\w, 0) -- (0.34*\w, 3) -- (0, 3) -- cycle;

        \filldraw[fill=red, ultra thick] (0.55*\w, 0) -- (\w, 0) -- (\w, 2.7) -- (0.66*\w, 2.7) -- cycle;

        \filldraw[fill=blue, ultra thick] (0, 3) -- (0.32*\w, 3) -- (0.26*\w, 6) -- (0, 6) -- cycle;

        \filldraw[fill=blue, ultra thick] (0.34*\w, 3) -- ++(0.32*\w, 0) -- (0.6*\w, 5) -- (0.36*\w, 5) -- cycle;

        \filldraw[fill=blue, ultra thick] (0.67*\w, 2.7) -- (\w, 2.7) -- (\w, 5.6) -- (0.74*\w, 5.6) -- cycle;

        \filldraw[fill=blue, ultra thick] (0.36*\w, 3) -- (0.63*\w, 3) -- (0.6*\w, 2.3) -- (0.4*\w, 2.3) -- cycle;

        \draw[very thick] (0, \h) -- (0, 0) -- (\w, 0) -- (\w, \h);
    \end{scope}

    \draw[|-|, ultra thick] (-0.5, 1.3) -- ++(0, 1) node[midway, anchor=east] {\(I_2\)};
    \draw[dashed] (-0.5, 1.3) -- ++(\w+1, 0);
    \draw[dashed] (-0.5, 2.3) -- ++(\w+1, 0);

    \draw[|-|, ultra thick] (-0.5, 5) -- ++(0, 1) node[midway, anchor=east] {\(I_3\)};
    \draw[dashed] (-0.5, 5) -- ++(\w+1, 0);
    \draw[dashed] (-0.5, 6) -- ++(\w+1, 0);

\end{tikzpicture}
            \caption{Schedule \(\sigma_2\)}
        \end{subfigure}%
    }

    \setcounter{subfigure}{0}

    \begin{minipage}[b][\ht\bigimage][s]{0.48\linewidth}

        \begin{subfigure}{\linewidth}
            \centering
            \begin{tikzpicture}

    %outline of the schedule
    \def\w{10}
    \def\h{3.2}

    %Sigma 1T
    \begin{scope}[fill opacity=0.7]

        \filldraw[fill=yellow, ultra thick] (\w, 0) --(\w, 3) -- (0.5*\w, 3) -- (0.4*\w, 0) -- cycle;

        \filldraw[fill=red, ultra thick] (0, 3) -- (0.4*\w, 3) -- (0.34*\w, 2) -- (0, 2) -- cycle;

        \filldraw[fill=blue, ultra thick] (0, 2) -- (0.32*\w, 2) -- (0.26*\w, 1) -- (0, 1) -- cycle;

        \draw[very thick] (0, \h) -- (0, 0) -- (\w, 0) -- (\w, \h);
    \end{scope}

    \draw[|-|, ultra thick] (-0.5, 0) -- ++(0, 1) node[midway, anchor=east] {\(I_1^T\)};

    \draw[dashed] (-0.5, 1) -- ++(\w+1, 0);

\end{tikzpicture}
            \caption{Schedule \(\sigma_1^T\)}
        \end{subfigure}%

        \vfill

        \begin{subfigure}{\linewidth}
            \centering
            \begin{tikzpicture}

    %outline of the schedule
    \def\w{10}
    \def\h{1.7}

    %Sigma 1B
    \begin{scope}[fill opacity=0.7]

        \filldraw[fill=yellow, ultra thick] (\w, 0) --(\w, 1.5) -- (0.4*\w, 1.5) -- (0.25*\w, 0) -- cycle;

        \draw[very thick] (0, \h) -- (0, 0) -- (\w, 0) -- (\w, \h);
    \end{scope}

    \draw[|-|, ultra thick] (-0.5, 0) -- ++(0, 1.5) node[midway, anchor=east] {\(I_1^B\)};
    \draw[dashed] (-0.5, 1.5) -- ++(\w+1, 0);

\end{tikzpicture}
            \caption{Schedule \(\sigma_1^B\)}
        \end{subfigure}%

    \end{minipage}\hfill % 
    \usebox{\bigimage}

    \caption{Constructed schedules with sparse intervals.}
    \Description{
        In this figure, we see three schedules $\sigma_1^B,\sigma_1^T$ and $\sigma_2$.
        $\sigma_1^T$ has a sparse interval $I_1^T$ where the machine usage is monotonically increasing above $I_1^T$ is a section where two jobs were executed parallel to each other.
        Also $\sigma_1^B$ has a sparse interval $I_1^B$ where the machine usage is monotonically increasing.
        Both intervals are accessible monotone.
        $\sigma_2$ has two sparse intervals $I_2$ and $I_3$. $I_2$ is not accessible and $I_3$ is accessible monotone.
    }
    \label{fig:schedules}
\end{figure}
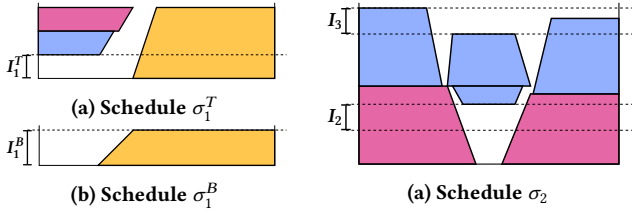

\subsection{Outline and Techniques}
\begin{algorithm}
    \caption{Outline of the Algorithm}
    \label{alg:outline}
    \begin{algorithmic}[1]
        \State Construct partial schedule \(\sigma_1\). %: containing all jobs in \(\mathcal{B}\)
        \Comment{\cref{alg:T1}}
        \State Construct partial schedule \(\sigma_2\). %: containing remaining jobs in \(\mathcal{M} \cup \mathcal{S}\)
        \Comment{\cref{alg:sigma_2_prime,alg:sigma2}}
        \State Combine \(\sigma_1\) and \(\sigma_2\) \Comment{\cref{sec:medium-jobs}}
        \State Greedily add tiny jobs \(\calT\)
        \Comment{\cref{sec:tiny-jobs}}
    \end{algorithmic}
\end{algorithm}

To simplify the presentation of our algorithm, we first assume that there are no tiny jobs, i.e., \(\calT = \emptyset\).
We show how to extend to the general case with tiny jobs in \cref{sec:tiny-jobs}.
% We will first prove the ratio for $\mathcal{T}=\emptyset$ where we have no tiny jobs and add those jobs then in $\cref{sec:tiny-jobs}$.
As outlined in \cref{alg:outline}, we sequentially construct partial schedules ($\sigma_1$ and $\sigma_2$).
Together, these place all big, medium, and small jobs.
During our analysis, we will show that these schedules have the following properties, where \(\sigma_1\) can be split into two parts, \(\sigma_1^B\) and \(\sigma_1^T\), which are scheduled on top of each other.

\begin{properties}[Constructed Schedules: $\sigma_1^B, \sigma_1^T, \sigma_2$] \leavevmode
    \label{def:named_schedules}
    \begin{itemize}
        \item $\sigma_1^B$ has a single accessible monotone sparse interval $I_1^B$ with $|I_1^B|< p_{\max}$ and $q(I_1^B) \subseteq (\frac{2}{3}, \frac34)$.
        \item $\sigma_1^T$ has a single accessible monotone sparse interval $I_1^T$ with $|I_1^T| < p_{\max}$ and $q(I_1^T) \subseteq (\frac12, \frac23]$.
        \item $\sigma_2$ has two sparse intervals \(I_2, I_3\), where $|I_2| < p_{\max}$ and $q(I_2) \subseteq (\frac23,\frac34)$ and \(I_3\) is accessible monotone with $|I_3|<p_{\max}$ and $q(I_3) \subseteq (\frac14,\frac13]\sqcup(\frac12,\frac34)$.
    \end{itemize}
\end{properties}

\begin{lemma}
    Given $\sigma_1^B, \sigma_1^T, \sigma_2$, we can create a schedule $\sigma$ with: \[h(\sigma) \leq \frac{4}{3}\opt + \pmax.\]
    \label{lem:combine_schedules}
\end{lemma}
\begin{proof}
    Let \(\sigma_1^B, \sigma_1^T, \sigma_2\) be as in \cref{def:named_schedules} and illustrated in \Cref{fig:schedules}.
    We aim to place these schedules on top of each other (while concurrently scheduling some sparse intervals) to create a new schedule \(\sigma\).

    First, create a schedule \(\sigma_1'\) by simply scheduling \(\sigma_1^B\) on top of \(\sigma_1^T\).
    To this end, note that the sparse interval \(I_3\) of \(\sigma_2\) has monotone machine utilization.
    Therefore, we can split \(I_3\) into two subintervals \(I_3^T\) and \(I_3^B\) such that \(q_{min}(I_3^T)\) is at most \(\frac{1}{3}\) and \(q_{min}(I_3^B)\) is at least \(\frac{1}{2}\).
    We now create a new schedule $\sigma$ by combining \(\sigma_1'\) and \(\sigma_2\).
    Due to the machine utilization of the sparse intervals \(I_1^T\) and \(I_3^T\), we can schedule them in a way that they interleave.
    Here, we schedule \(\sigma_2\) at time \(0\) and \(\sigma_1'\) at time \(h(\sigma_2)-\min(|I_1^T|,|I_3^T|)\).
    We can schedule those sparse intervals simultaneously since $q_{max}(I_3^T)+q_{max}(I_1^T)\leq\frac13+\frac23=1$.
    Note that in general it does not matter whether the accessible sparse intervals of a schedule $\sigma'$ are at the top or the bottom, since we can just rotate the schedule (that is change $\sigma'(j)$ to $h(\sigma')-p_j-\sigma'(j)$ for every job).

    To analyze the area deficit of the new schedule \(\sigma\), we first observe that by definition only sparse intervals contribute to the area deficit.
    Also, during the time where \(I_1^T\) and \(I_3^T\) interleave, the combined machine utilization is \(q_{min}(I_1^T) + q_{min}(I_3^T) > \frac{1}{2} + \frac{1}{4} = \frac{3}{4}\).
    This interleaved region does not contribute to the area deficit.
    Depending on which interval is larger, we have two cases:

    \case{1} \(|I_1^T| \leq |I_3^T|\).
    In this case, \(I_1^T\) is fully absorbed into the overlap.
    Using the definition of \(\deficit(\sigma)\), we get:
    \begin{align*}
        \deficit(\sigma)
         & = \int_{0}^{h(\sigma)} \max(0,\frac{3}{4} - q(\sigma, \tau)) d\tau                                                           \\
         & \leq \int_{I_1^B} \left(\frac{3}{4} - q_{min}(I_1^B)\right) d\tau + \int_{I_2} \left(\frac{3}{4} - q_{min}(I_2)\right) d\tau \\
         & \quad +\int_{I_3}\left(\frac{3}{4} - q_{min}(I_3)\right) d\tau                                                               \\
    \end{align*}
    With \(|I_1^B|, |I_2|, |I_3| \leq \pmax\), we have:
    \begin{align*}
        \deficit(\sigma)
         & \leq \int_{I_1^B} \left(\frac{3}{4} - \frac{2}{3}\right) d\tau + \int_{I_2} \left(\frac{3}{4} - \frac{2}{3}\right) d\tau + \int_{I_3} \left(\frac{3}{4} - \frac{1}{2}\right) d\tau \\
         & = \left(\frac{3}{4}-\frac{2}{3}\right)|I_1^B| + \left(\frac{3}{4}-\frac{2}{3}\right)|I_2| + \left(\frac{3}{4}-\frac{1}{2}\right)|I_3|                                              \\
         & \leq \pmax
    \end{align*}

    \case{2} \(|I_1^T| > |I_3^T|\).
    In this case, \(I_3^T\) is fully absorbed into the overlap.
    We analogously bound the area deficit as follows:
    \begin{align*}
         & \left(\frac{3}{4}-\frac{2}{3}\right)|I_1^B| + \left(\frac{3}{4}-\frac{2}{3}\right)|I_2| + \left(\frac{3}{4}-\frac{1}{2}\right)|I_3^B| + \left(\frac{3}{4}-\frac{1}{2}\right)|I_1^T| \\
         & < \frac{3}{4}\pmax.
    \end{align*}

    The final statement then follows from \cref{obs:missing_area_to_ratio}.
\end{proof}

\subsection{Creating \texorpdfstring{$\sigma_1^B, \sigma_1^T$}{sigma1B, sigma1T} and Scheduling all Big Jobs}
\label{sec:big-jobs}
In this section, we aim to create $\sigma_1^B$ and $\sigma_1^T$ from \cref{def:named_schedules}
and thereby schedule all big jobs.
We will create a schedule \(\sigma_1\) using \cref{alg:T1}.
This schedule has a single sparse interval \(I_1\); however, the height of this interval might be larger than \(\pmax\).
We show that we may assume w.l.o.g.\ that the height of this interval is at most \(\pmax \), assuming \(\mathcal{T} = \emptyset \).
Then we show how to split this schedule into two schedules \(\sigma_1^B\) and \(\sigma_1^T\) with the properties as defined in \cref{def:named_schedules}.

We construct \(\sigma_1\) as detailed in \cref{alg:T1}. First, we stack all big jobs (\(\calB\)) in non-increasing order of their widths. Next, we fill the remaining capacity from the top downwards, starting at \(\tau = h(\sigma_1)\). We repeatedly select the widest unscheduled job \(j \in \calM \sqcup \calS\) that fits entirely within \([\tau - p_j, \tau)\)—requiring \(q_j \le 1 - q(\sigma_1, \tau - p_j)\)—and does not cross time \(0\) (this means that $\tau-p_j$ is always at least $0$). We continue this process until no further jobs can be placed.

\begin{algorithm}[ht]
    \caption{\(\sigma_1\): Scheduling Jobs in \(\calB\) and Top-Down Filling}
    \label{alg:T1}
    \begin{algorithmic}[1]
        \State Stack all \(j \in \calB\) in non-increasing order of width (\(q_j\)).
        \State \(\tau \gets h(\sigma_1)\)
        \While{there is an unscheduled \(j \in \calM \sqcup \calS\) fitting in \([\tau - p_j, \tau)\) with \(\tau - p_j \ge 0\)}
        \State Let \(j\) be the widest such job.
        \State Schedule \(j\) at time \(\tau - p_j\).
        \State \(\tau \gets \tau - p_j\)
        \EndWhile
    \end{algorithmic}
\end{algorithm}
Note that at any time \(\tau\) in $\sigma_1$ where two jobs are scheduled simultaneously (one in \(\calB\) and one in \(\calM \sqcup \calS\)), the machine utilization is at least $\frac14+\frac12=\frac34$.
\begin{observation}[\(\sigma_1\)-schedule]
    \label{obs:sigma1}
    The schedule \(\sigma_1\) constructed by \cref{alg:T1} uses at least \(\frac{3}{4}\) of the machines at any time, except for a single accessible sparse interval \(I_1\).
\end{observation}

\begin{figure}[ht]
    \centering
    \begin{subfigure}{0.23\textwidth}
        \centering
        \begin{tikzpicture}
    %outline of the schedule
    \def\w{10}
    \def\h{4.5}
    \def\xright{\w} % fixed right x-bound
    \def\xleft{0} % fixed left x-bound

    \begin{scope}[fill opacity=0.7]

        %Big jobs
        \begin{scope}[fill=yellow]

            \pgfmathsetmacro{\ycur}{0} % Initial y value
            % Draw stacked rectangles
            \expandafter\stackrectright\jobBA
            \expandafter\stackrectright\jobBB
            \expandafter\stackrectright\jobBC
            \expandafter\stackrectright\jobBD
            \expandafter\stackrectright\jobBE
            \pgfmathsetmacro{\scheduleheight}{\ycur}

            \draw[dashed, ultra thick] (\w, \scheduleheight) -- (-0.5, \scheduleheight) node[anchor=east] {\(h(\sigma_1\))};
            %\draw[decorate, decoration={brace, amplitude=10pt}] (\w, \scheduleheight) -- (\w, 0) node[midway, xshift=3em] {\(\sigma_1\)};
            \draw[|-|, ultra thick] (-0.5, 0) -- (-0.5, \scheduleheight) node[midway, anchor=east] {\(I_1\)};
        \end{scope}

        \draw[very thick] (0, \h) -- (0, 0) -- (\w, 0) -- (\w, \h);
    \end{scope}

\end{tikzpicture}
        \subcaption{during \cref{alg:T1}}
        \label{fig:sigma1_intermediate}
    \end{subfigure}
    \hfill
    \begin{subfigure}{0.23\textwidth}
        \centering
        \begin{tikzpicture}
    %outline of the schedule
    \def\w{10}
    \def\h{4.5}
    \def\xright{\w} % fixed right x-bound
    \def\xleft{0} % fixed left x-bound

    \begin{scope}[fill opacity=0.7]

        %Big jobs
        \begin{scope}[fill=yellow]

            \pgfmathsetmacro{\ycur}{0} % Initial y value
            % Draw stacked rectangles
            \expandafter\stackrectright\jobBA
            \expandafter\stackrectright\jobBB
            \expandafter\stackrectright\jobBC
            \expandafter\stackrectright\jobBD
            \expandafter\stackrectright\jobBE
            \pgfmathsetmacro{\scheduleheight}{\ycur}

            %Medium jobs
            \begin{scope}[fill=red]
                %\expandafter\stackrectleft\jobMA{down}
                \expandafter\stackrectleft\jobMG{down}

                %Small jobs
                \begin{scope}[fill=blue]
                    \expandafter\stackrectleft\jobSG{down}
                \end{scope}
            \end{scope}

            \draw[dashed, ultra thick] (\w, \scheduleheight) -- (-0.5, \scheduleheight) node[anchor=east] {\(h(\sigma_1\))};

            \begin{scope}[color=black]

                \draw[dashed, thick] ($(leftbottom)+(-0.5, 0)$) node[color=black, anchor=east] {\(\tau\)} -- +(\w+0.5, 0);
                \draw[|-|, ultra thick] (-0.5, 0) -- ($(leftbottom) + (-0.5,0)$) node[midway, anchor=east] {\(I_1\)};
            \end{scope}
        \end{scope}

        \draw[very thick] (0, \h) -- (0, 0) -- (\w, 0) -- (\w, \h);
    \end{scope}

\end{tikzpicture}
        \subcaption{after \cref{alg:T1}}
        \label{fig:sigma1}
    \end{subfigure}
    \hfill
    \caption{Schedule \(\sigma_1\)}
    \Description{
        In this figure we see two schedules.
        On the left side we see $\sigma_1$ after the tall jobs were placed and on the right side we see $\sigma_1$ after the whole algorithm.
        We see that $\sigma_1$ (on the right side) has an accessible monotone gap $I_1$ which might be larger than $p_{max}$.
    }
\end{figure}
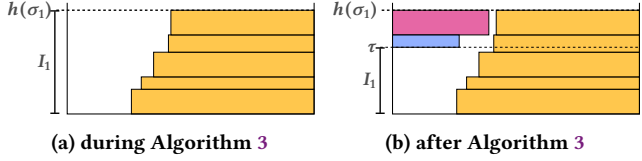
Firstly, there is one important observation, regarding the placement of small and medium jobs.
By construction, the machine requirement of those jobs is non-decreasing in time.
This implies that the placement of a job \(j\) is not affected by the presence of any job \(k\) with \(q_k < q_j\) this implies:
\begin{align} \label{obs:sigma1_subset_removal}
    dom(\sigma_1(\mathcal{J}\setminus \mathcal{S}))=dom(\sigma_1(\mathcal{J}))\setminus \mathcal{S}.
\end{align}

In addition, note that $I_1$ might be of arbitrary height.
Now, we show that we can reduce any such instance to an instance where $I_1$ has height at most $\pmax$, when \(\calT = \emptyset \).
We do this by removing some jobs from the instance and showing that we can add them back while maintaining our approximation guarantee.

\begin{restatable}{lemma}{instancemod}
    \label{lem:instance-mod}
    If \(\calT = \emptyset \), any instance $\mathcal{J}$ can be reduced to an instance $\mathcal{J}'\subseteq \mathcal{J}$ where \cref{alg:T1} creates a schedule with a single accessible sparse interval \(I_1\) of height \(|I_1| < \pmax\). In this context, "reducing" implies that reinserting those jobs from $\mathcal{J}\setminus \mathcal{J}'$ into the schedule (with \cref{alg:T1}) does not increase the approximation ratio.
\end{restatable}
\def\jmin{{j_-}}
\def\jmax{{j_+}}
\begin{proof}
    Assume \cref{alg:T1} returns a schedule \(\sigma_1\) with $|I_1|>\pmax$.
    We show how to iteratively remove certain jobs from the instance that can be easily reinserted later, while maintaining the approximation guarantee.

    Let \(\jmin \) be the thinnest job processed in \(\sigma_1\), and let \(q_\jmin\) be its width.
    Similarly, let \(\jmax \) be the widest job processed in \(\sigma_1\), and let \(q_\jmax \) be its width.

    The following two claims show that, while \(I_1\) has height larger than \(\pmax\), one of these two jobs can be removed from the instance.

    \begin{claim}
        \label{lem:remove-big-jobs}
        Suppose $|I_1|\geq \pmax$ and \(q_\jmax > 1 - q_\jmin\).
        Given any schedule \(\sigma\) for all jobs in \(\calJ\),
        we can construct a schedule \(\sigma' \) of height \( h(\sigma') = h(\sigma) - p_\jmax \) for all jobs in \(\calJ \setminus \{\jmax \} \).
    \end{claim}
    \begin{subproof}[Proof of \cref{lem:remove-big-jobs}]
        We first argue, that there exists no job that can be processed concurrently to \(\jmax \).
        For the sake of argument, assume there is a job \(j \in \calJ\) with \(q_\jmax + q_j \leq 1\).

        \case{1} \(j\) is scheduled in \(\sigma_1\).
        Since \(\jmin \) is the thinnest job in \(\sigma_1\): \(q_\jmax + q_j \geq q_\jmax + q_\jmin > 1\). A contradiction to the choice of \(j\).

        \case{2} \(j\) is not scheduled in \(\sigma_1\).
        Then there is a sparse time-interval in the schedule with height at least \(|I_1| \geq \pmax \geq p_j\) and machine usage at most $q_\jmax$ and during any time in the interval.
        Since $j$ fits in this interval ($q_\jmax+q_j\leq 1)$, \cref{alg:T1} would have scheduled $j$.

        Given any schedule \(\sigma \) for all jobs in \(\calJ \), we can remove the widest job \(\jmax \), if \(|I_1| \geq \pmax \).
        Since no other job can be processed concurrently to \(\jmax \), all jobs \(j \in \calJ\) with \(\sigma(j) > \sigma(\jmax)\) can be scheduled exactly \(q_\jmax\) earlier.
        The new schedule has a height of \(h(\sigma) - p_\jmax\).
    \end{subproof}

    Note that this also holds for any optimal schedule.
    Therefore,
    \begin{align}
        \label{eq:remove-big-job}
        \opt(\calJ \setminus \{\jmax\}) = \opt(\calJ) - p_\jmax.
    \end{align}
    Next, we show that an \(\alpha\)-approximation algorithm for all jobs in \(\calJ \setminus \{\jmax\}\) also yields an \(\alpha\)-approximation for all jobs in \(\calJ\).
    \begin{corollary}
        \label{cor:instance_reduction_big_jobs}
        Given \(\alpha > 1\) and a schedule \(\sigma'\) with makespan at most \(\alpha \opt(\calJ \setminus \{\jmax\})\)
        for all jobs in \(\calJ \setminus \{\jmax\}\), we can construct a schedule \(\sigma\) with makespan at most
        \(\alpha \opt(\calJ)\) for all jobs in \(\calJ\).
    \end{corollary}
    \begin{subproof}[Proof of \cref{cor:instance_reduction_big_jobs}]
        We now construct \(\sigma \) from \(\sigma'\) by simply scheduling \(\jmax \) separately.
        This results in a height of \(h(\sigma) = h(\sigma') + p_\jmax \).
        With \cref{eq:remove-big-job}, we get:
        \begin{align*}
            h(\sigma) & = h(\sigma') + p_\jmax                                                                  \\
                      & \leq \alpha\left(\opt(\calJ\setminus\{\jmax\})\right)+ p_\jmax                          \\
                      & \leq \alpha(\opt(\calJ) - p_\jmax) + p_\jmax                                            \\
                      & < \alpha \opt(\calJ)                                           &  & \alpha > 1 \qedhere
        \end{align*}
    \end{subproof}

    By applying \cref{lem:remove-big-jobs} iteratively, we may assume that \(q_\jmax \leq 1 - q_\jmin\) since any such job can be re-added while maintaining the approximation guarantee.

    \begin{claim}
        \label{lem:instance_reduction_small_jobs}
        Suppose \(q_\jmax \leq 1 - q_\jmin\) and \(|I_1| \geq \pmax\).
        Given any schedule \(\sigma'\) returned by \cref{alg:T1} for all jobs in \(\calJ \setminus \{\jmin\}\), we can insert \(\jmin\) without increasing the height of the schedule.
    \end{claim}
    \begin{subproof}[Proof of \cref{lem:instance_reduction_small_jobs}]
        Let \(\sigma_1\) be the schedule constructed by \cref{alg:T1} for all jobs in \(\calJ\).
        By construction, \cref{alg:T1} first stacks all big jobs in \(\calB\). This stack is identical whether \(\jmin\) is present or not.
        It then iteratively schedules jobs from \(\calM \cup \calS\) top-down, always selecting the \emph{widest} fitting job.

        Since \(\jmin\) is the thinnest job scheduled in \(\sigma_1\), any job \(j \in dom(\sigma_1) \setminus \{\jmin\}\) is at least as wide as \(\jmin\).
        Because \cref{alg:T1} prioritizes wider jobs, the absence of \(\jmin\) does not alter the placement of any job scheduled before \(\jmin\) during the top-down packing (as implied by \cref{obs:sigma1_subset_removal}).

        Now consider running \cref{alg:T1} on \(\calJ \setminus \{\jmin\}\). The sequence of placements remains identical until the point where \(\jmin\) would have been scheduled. Instead of \(\jmin\), the algorithm searches for the next widest fitting job among the unscheduled jobs \(\calU = \calJ \setminus dom(\sigma_1)\).
        We argue that no \(u \in \calU\) can be placed:

        \case{1} \(q_u > q_\jmin\).
        Job \(u\) was already evaluated before \(\jmin\) and rejected because it did not fit.
        Since the available space at that moment is unchanged by \(\jmin\)'s absence, \(u\) still does not fit.

        \case{2} \(q_u \leq q_\jmin\).
        Job \(u\) would have fit in the sparse interval \(I_1\) that remains below \(\jmin\). We know \(|I_1| \geq \pmax \geq p_u\) and, since \(q_\jmax \leq 1 - q_\jmin\), the available width in \(I_1\) is at least \(1 - q_\jmax \geq q_\jmin \geq q_u\). Because \(u\) was not scheduled in \(\sigma_1\) despite fitting, no such job exists in \(\calU\).

        Therefore, running \cref{alg:T1} on \(\calJ \setminus \{\jmin\}\) produces a schedule \(\sigma'\) that is exactly \(\sigma_1 \setminus \{\jmin\}\).
        We can construct the schedule \(\sigma\) for \(\calJ\) simply by taking \(\sigma'\) and inserting $\jmin$ into its original position.
        This yields a schedule with \(h(\sigma) = h(\sigma')\), proving the claim.
    \end{subproof}
    To prove \cref{lem:instance-mod}, we iteratively apply \cref{lem:remove-big-jobs,lem:instance_reduction_small_jobs} until \(|I_1| < \pmax \).
    We first note that each application of these claims removes at least one job.
    Consequently, the procedure is guaranteed to terminate.
    We have shown that by reintroducing all the jobs removed via \cref{lem:remove-big-jobs}, we do not increase the approximation ratio.
    Then, reinserting all jobs removed by \cref{lem:instance_reduction_small_jobs} does not increase the height of the schedule.
    \qedhere

\end{proof}

\begin{remark}
    Since \cref{lem:instance-mod} can only be applied when \(\calT = \emptyset\), we only use it in \cref{sec:medium-jobs}.
    For the generalization in \cref{sec:tiny-jobs} we will not use this assumption.
\end{remark}

We now partition $\sigma_1$ into $\sigma_1^B$ and $\sigma_1^T$.
Since we will also need the partitioned schedule when placing the tiny jobs, we prove a more general version of this lemma.
For the following lemma, we assume that not all small jobs have been placed in $\sigma_1$.
The opposite case is easier and will be evaluated directly later (\cref{obs:assump_sigma_2}).
\begin{restatable}{lemma}{sigmapartition}
    \label{lem:sigma1_partition}
    If \(dom(\sigma_1) \cap \calS \neq \calS\), we can find some \(\tau\) such that the two schedules defined by
    \begin{equation}
        \label{eq:g_1_partition}
        \sigma_1^B(j) \coloneqq \begin{cases}
            \sigma_1(j) & \text{if } \sigma_1(j) \leq \tau \\
            \bot        & \text{otherwise}
        \end{cases}
        \hspace{1.5em}
        \sigma_1^T(j) \coloneqq \begin{cases}
            \sigma_1(j) & \text{if } \sigma_1(j) > \tau \\
            \bot        & \text{otherwise}
        \end{cases}
    \end{equation}
    \begin{itemize}
        \item satisfy \cref{def:named_schedules} exactly if \(\calT=\emptyset\) or
        \item satisfy \cref{def:named_schedules} except for the \(|I_1^B|<p_{max}\) constraint.
    \end{itemize}

    where \(\bot\) denotes that the job is not scheduled in the partial schedule.
\end{restatable}

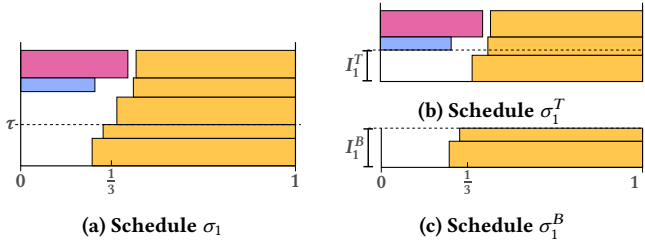
\begin{figure}[t]
    \centering

    \begin{subfigure}[b]{.47\columnwidth}
        \centering
        \begin{tikzpicture}

    %outline of the schedule
    \def\w{10}
    \def\h{4.5}
    \def\xright{\w} % fixed right x-bound
    \def\xleft{0} % fixed left x-bound

    %Sigma 1
    \begin{scope}[fill opacity=0.7]

        %Big jobs
        \begin{scope}[fill=yellow]

            \pgfmathsetmacro{\ycur}{0} % Initial y value
            % Draw stacked rectangles
            \expandafter\stackrectright\jobBA
            \expandafter\stackrectright\jobBB
            \pgfmathsetmacro{\cutheight}{\ycur}
            \expandafter\stackrectright\jobBC
            \expandafter\stackrectright\jobBD
            \expandafter\stackrectright\jobBE

            %Medium jobs
            \begin{scope}[fill=red]
                %\expandafter\stackrectleft\jobMA{down}
                \expandafter\stackrectleft\jobMG{down}

                %Small jobs
                \begin{scope}[fill=blue]
                    \expandafter\stackrectleft\jobSG{down}
                \end{scope}
            \end{scope}

            \draw[dashed] (-0.2, \cutheight) node[xshift=-0.5em] {\(\tau\)} -- (\w+0.2, \cutheight);
            %\draw[decorate, decoration={brace, amplitude=10pt}, color=gapcolor] (0, 0) -- (leftbottom) node[midway, xshift=-2em, yshift=-0.5em] {\(G_1\)};
            % \draw[decorate, decoration={brace, amplitude=10pt}] (\w, \cutheight) -- (\w, 0) node[midway, xshift=3em] {\(\sigma_1^B\)};
            % \draw[decorate, decoration={brace, amplitude=10pt, mirror}] (\w, \cutheight) -- (righttop) node[midway, xshift=3em] {\(\sigma_1^T\)};
            % \draw[decorate, decoration={brace, amplitude=10pt, mirror}, color=gapcolor] (0, \cutheight) -- +(0, -0.4) node[midway, xshift=-3em, yshift=-1em] {\(G_1^B\)};
            % \draw[decorate, decoration={brace, amplitude=10pt}, color=gapcolor] (0, \cutheight) -- (leftbottom) node[midway, xshift=-3em] {\(G_1^T\)};
        \end{scope}

        \draw (0, 0) -- +(0, -0.2) node[anchor=north] {\(0\)};
        \draw (3.3, 0.2) -- +(0, -0.2) node[anchor=north] {\(\frac{1}{3}\)};
        \draw (\w, 0) -- +(0, -0.2) node[anchor=north] {\(1\)};
        \draw[very thick] (0, \h) -- (0, 0) -- (\w, 0) -- (\w, \h);
    \end{scope}

\end{tikzpicture}
        \caption{Schedule \(\sigma_1\)}
    \end{subfigure}%
    \hfill%
    \begin{minipage}[b][][s]{.47\columnwidth}
        \begin{subfigure}{\linewidth}
            \centering
            \begin{tikzpicture}

    %outline of the schedule
    \def\w{10}
    \def\h{3}
    \def\xright{\w} % fixed right x-bound
    \def\xleft{0} % fixed left x-bound

    %Sigma 1
    \begin{scope}[fill opacity=0.7]

        %Big jobs
        \begin{scope}[fill=yellow]

            \pgfmathsetmacro{\ycur}{0} % Initial y value
            % Draw stacked rectangles
            \pgfmathsetmacro{\cutheight}{\ycur}
            \expandafter\stackrectright\jobBC
            \expandafter\stackrectright\jobBD
            \expandafter\stackrectright\jobBE

            %Medium jobs
            \begin{scope}[fill=red]
                %\expandafter\stackrectleft\jobMA{down}
                \expandafter\stackrectleft\jobMG{down}

                %Small jobs
                \begin{scope}[fill=blue]
                    \expandafter\stackrectleft\jobSG{down}
                \end{scope}
            \end{scope}

            %\draw[decorate, decoration={brace, amplitude=10pt}, color=gapcolor] (0, 0) -- (leftbottom) node[midway, xshift=-2em, yshift=-0.5em] {\(G_1\)};
            % \draw[decorate, decoration={brace, amplitude=10pt}] (\w, \cutheight) -- (\w, 0) node[midway, xshift=3em] {\(\sigma_1^B\)};
            % \draw[decorate, decoration={brace, amplitude=10pt, mirror}] (\w, \cutheight) -- (righttop) node[midway, xshift=3em] {\(\sigma_1^T\)};
            % \draw[decorate, decoration={brace, amplitude=10pt, mirror}, color=gapcolor] (0, \cutheight) -- +(0, -0.4) node[midway, xshift=-3em, yshift=-1em] {\(G_1^B\)};
            \begin{scope}[color=black]
                \draw[dashed, thick] ($(leftbottom)+(-0.5, 0)$)-- +(\w+0.5, 0);
                \draw[|-|, ultra thick] (-0.5, \cutheight) -- ($(leftbottom)+(-0.5, 0)$) node[midway, anchor=east] {\(I_1^T\)};
            \end{scope}
        \end{scope}

        % \draw (0, 0) -- +(0, -0.2) node[anchor=north] {\(0\)};
        % \draw (3.3, 0.2) -- +(0, -0.2) node[anchor=north] {\(\frac{1}{3}\)};
        % \draw (\w, 0) -- +(0, -0.2) node[anchor=north] {\(1\)};
        \draw[very thick] (0, \h) -- (0, 0) -- (\w, 0) -- (\w, \h);
    \end{scope}

\end{tikzpicture}
            \caption{Schedule \(\sigma_1^T\)}
        \end{subfigure}
        \vfill
        \begin{subfigure}{\linewidth}
            \centering
            \begin{tikzpicture}

    %outline of the schedule
    \def\w{10}
    \def\h{1.5}
    \def\xright{\w} % fixed right x-bound
    \def\xleft{0} % fixed left x-bound

    %Sigma 1
    \begin{scope}[fill opacity=0.7]

        %Big jobs
        \begin{scope}[fill=yellow]

            \pgfmathsetmacro{\ycur}{0} % Initial y value
            % Draw stacked rectangles
            \expandafter\stackrectright\jobBA
            \expandafter\stackrectright\jobBB
            \pgfmathsetmacro{\cutheight}{\ycur}

            \begin{scope}[color=black]
                \draw[dashed, thick] (-0.5, \cutheight) -- (\w, \cutheight);
                \draw[|-|, ultra thick] (-0.5, \cutheight) -- (-0.5, 0) node[midway, anchor=east] {\(I_1^B\)};

            \end{scope}
        \end{scope}

        \draw (0, 0) -- +(0, -0.2) node[anchor=north] {\(0\)};
        \draw (3.3, 0.2) -- +(0, -0.2) node[anchor=north] {\(\frac{1}{3}\)};
        \draw (\w, 0) -- +(0, -0.2) node[anchor=north east] {\(1\)};
        \draw[very thick] (0, \h) -- (0, 0) -- (\w, 0) -- (\w, \h);
    \end{scope}

\end{tikzpicture}
            \caption{Schedule \(\sigma_1^B\)}
        \end{subfigure}
    \end{minipage}

    \caption{Schedule \(\sigma_1\) with partitioning into \(\sigma_1^B\) and \(\sigma_1^T\)}
    \Description{
        In this figure we see on the left side $\sigma_1$ and a horizontal line $\tau$ which is the first time where a job has height smaller than $\frac23$.
        It is evident that $\tau$ must be between two jobs and not dividing one.
        On the right side of this we see two schedules $\sigma_1^T$ and $\sigma_1^B$ which were created by taking all jobs above $\tau$ and below $\tau$ respectively.
    }
\end{figure}

\begin{proof}
    For that we consider $\sigma_1$ as given in \cref{alg:T1}.
    Now denote by $\tau:=\min\{\tau'\in [0,h(\sigma_1)]:q(\sigma_1,\tau')<\frac{2}{3}\}$ the first time in $\sigma_1$,
    where the machine usage is smaller than $\frac{2}{3}$.
    If no such time exists, then we define $\tau\coloneq h(\sigma_1)$.
    Note that if $\tau =h(\sigma_1)$, then $\sigma_1^B=\sigma_1$ which proves the claim for this case.

    We now consider the following two remaining cases based on the distance between $\tau$ and the end \(b\) of the sparse interval $I_1 = [a, b)$.
    Note that the difference $b-\tau\eqqcolon |I_1^T|$ is always non-negative.

    \case{1} $b-\tau\geq \pmax$.
    In this case, we have a section of height at least $\pmax$ in $I_1$ where the machine usage is at most $\frac{2}{3}$.
    Therefore, we do not have any remaining small jobs after the application of \cref{alg:T1}, since they would have been placed in $\sigma_1$.
    This is a contradiction to the assumption \(dom(\sigma_1) \cap \calS \neq \calS\).

    \case{2} $b-\tau<\pmax$.
    In this case $\sigma_1^B$ and $\sigma_1^T$ fulfill all conditions of \cref{def:named_schedules} except for $|I_1^B|<\pmax$.
    Additionally, if $\calT\neq \emptyset$ then $|I_1^B|=|I_1|-|I_1^T|<\pmax$ by \cref{lem:instance-mod}.

    The two cases prove the claim.
\end{proof}

\subsection{Scheduling the Remaining Medium and Small Jobs}
\label{sec:medium-jobs}
In this section, we consider only those jobs, that have not yet been scheduled by \cref{alg:T1}.
The remaining instance does not contain any big jobs, i.e., for any remaining job $j$ it holds that \( \frac{1}{4} < q_j \leq \frac{1}{2}.\)

We build the second schedule by first constructing an intermediate schedule \(\sigma_2'\) that contains all (remaining) jobs in \(\calM\).
Scheduling the medium jobs is done by the simple procedure given in \cref{alg:sigma_2_prime}:
Sort all jobs in \(\calM\) in non-increasing order of machine requirements, then schedule these jobs via \LS on two stacks.
The resulting schedule is illustrated in \Cref{fig:sigma2_prime}.
In the following, denote the height difference of the two constructed stacks as \(|I_2'|\).
The next lemma gives an upper bound on the height of \(\sigma_1\) and \(\sigma_2'\).

\tikzset{font=\fontsize{28}{10}\selectfont\bfseries\boldmath}
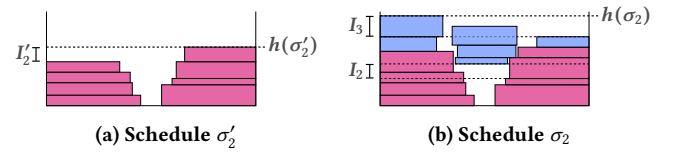
\begin{figure}[ht]
    \centering
    \begin{subfigure}{0.48\columnwidth}
        \centering
        \begin{tikzpicture}
    %outline of the schedule
    \def\w{10}
    \def\h{4.5}
    \def\xright{\w} % fixed right x-bound
    \def\xleft{0} % fixed left x-bound

    %Sigma 2'
    \begin{scope}[fill opacity=0.7]
        %Medium jobs
        \begin{scope}[fill=red]

            % Right stack
            \pgfmathsetmacro{\ycur}{0}
            \expandafter\stackrectright\jobMA
            \expandafter\stackrectright\jobMD
            \stackrectright{3.8}{0.8}
            \stackrectright{3.4}{0.7}

            \draw[dashed, ultra thick] (0, \ycur) -- (\w + 0.5, \ycur) node[anchor=west, yshift=0.0em] {\(h(\sigma_2')\)};
            \pgfmathsetmacro{\righttop}{\ycur}
            % Left stack
            \pgfmathsetmacro{\ycur}{0}
            \expandafter\stackrectleft\jobMB{up}
            \expandafter\stackrectleft\jobMC{up}
            \stackrectleft{4}{0.5}{up}
            \stackrectleft{3.5}{0.5}{up}

            \draw[|-|, ultra thick, color=black] (-0.5, \ycur) -- (-0.5, \righttop) node[midway, xshift=-2em, yshift=-0.0em] {\(I_2'\)};

        \end{scope}
        \draw[very thick] (0, \h) -- (0, 0) -- (\w, 0) -- (\w, \h);
    \end{scope}

\end{tikzpicture}
        \subcaption{Schedule \(\sigma_2'\)}
        \label{fig:sigma2_prime}
    \end{subfigure}
    \hfill
    \begin{subfigure}{0.48\columnwidth}
        \centering
        \begin{tikzpicture}
    %outline of the schedule
    \def\w{10}
    \def\h{4.5}
    \def\xright{\w} % fixed right x-bound
    \def\xleft{0} % fixed left x-bound

    %Sigma 2
    \begin{scope}[fill opacity=0.7]
        %Medium jobs
        \begin{scope}[fill=red]

            % Right stack
            \pgfmathsetmacro{\ycur}{0}
            \expandafter\stackrectright\jobMA
            \expandafter\stackrectright\jobMD
            \stackrectright{3.8}{1}
            \filldraw[fill=blue] (lefttop) ++(0.3, 0) rectangle ++(-2.8, 0.6) node (middle) {};
            \expandafter\stackrectright\jobMF

            \begin{scope}[fill=blue]
                \expandafter\stackrectright\jobSF
            \end{scope}

            % Left stack
            \pgfmathsetmacro{\ycur}{0}
            \expandafter\stackrectleft\jobMB{up}
            \expandafter\stackrectleft\jobMC{up}
            \stackrectleft{4}{0.5}{up}
            \filldraw[fill=blue] (righttop) ++(-0.4, 0.4) rectangle ++(2.5, 0.3);
            \expandafter\stackrectleft\jobME{up}
            \begin{scope}[fill=blue]
                \expandafter\stackrectleft{2.7}{0.7}{up}
                \expandafter\stackrectleft{3}{1}{up}

                %\draw[dashed, ultra thick] (\w, \ycur) -- (-0.5, \ycur) node[anchor=east,yshift=0.5em] {\(h(\sigma_2)\)};
                \draw[dashed, ultra thick] (0, \ycur) -- (\w + 0.5, \ycur) node[anchor=west, yshift=0.0em] {\(h(\sigma_2)\)};
                
                \begin{scope}[color=black]
                    %\draw[decorate, decoration={brace, amplitude=10pt}] (\w, \ycur) -- +(0, -0.5) node[midway, xshift=3em] {\(G_3^T\)};
                    \draw[|-|, ultra thick] (-0.5, \ycur-1) -- +(0, 1) node[midway, anchor=east, xshift=-0.5em, yshift=-0.0em] {\(I_3\)};
                    % \draw[dashed, thick] (\w, \ycur-0.5) -- (0, \ycur-0.5);
                    %\draw[decorate, decoration={brace, amplitude=10pt}] (\w, \ycur-0.5) -- +(0, -0.5) node[midway, xshift=3em] {\(G_3^B\)};
                    \draw[dashed, thick] (\w, \ycur-1) -- (0, \ycur-1);

                    \draw[dashed, thick] (0, 2) -- +(\w, 0);
                    \draw[|-|, ultra thick] (-0.5, 1.3) -- +(-0, 0.7) node[midway, anchor=east, xshift=-0.5em, yshift=-0.0em] {\(I_2\)};
                    \draw[dashed, thick] (0, 1.3) -- +(\w, 0);

                \end{scope}
            \end{scope}

            %middle stack
            \filldraw[fill=blue] (middle) ++(-0.25, 0) rectangle ++(3.1, 0.9);

        \end{scope}

        \draw[very thick] (0, \h) -- (0, 0) -- (\w, 0) -- (\w, \h);
    \end{scope}

\end{tikzpicture}
        \subcaption{Schedule \(\sigma_2\)}
        \label{fig:sigma2}
    \end{subfigure}
    \caption{Schedule \(\sigma_2\) after \cref{alg:sigma_2_prime,alg:sigma2}}
    \Description{
        On the left side of the figure we see $\sigma_2'$ which is an intermediate schedule which was created before adding the small jobs and consists of two stacks of medium jobs.
        We see that $I_2'$ is just the height difference between the two stacks.
        On the right side of the figure is $\sigma_2$ which is $\simga_2'$ with the remaining small jobs which were greedily placed on top and then moved up to create to sparse intervals..
        We see that this placement leads to the creation of two sparse intervals $I_2$ below the starting time of the first small jobs and $I_3$ which is the height difference between the three stacks which were create by placing the small jobs greedily.
    }
\end{figure}

\begin{algorithm}[ht]
    \caption{\(\sigma_2'\): Scheduling the Remaining Jobs in \(\calM\)}
    \label{alg:sigma_2_prime}
    \begin{algorithmic}[1]
        \State Place all jobs \(j \in \calM\) via \LS on two stacks.
    \end{algorithmic}
\end{algorithm}

\begin{lemma}\label{lem:LB opt}
    The height of the schedules \(\sigma_1\) and \(\sigma_2'\), constructed by \cref{alg:T1,alg:sigma_2_prime}, satisfies:
    \begin{align*}
        h(\sigma_1) + h(\sigma_2') \leq \opt + \frac{|I_1|}{2} + \frac{|I_2'|}{2}
    \end{align*}
\end{lemma}

\begin{proof}
    For simplicity, we start with the assumption \(\calS = \emptyset\).
    Consider now the schedule \(\sigma_1\) and the intermediate schedule \(\sigma_2'\).
    Denote by \(P \coloneqq \sum_{j\in \mathcal{M}\sqcup\mathcal{B}}p_j\) the summed processing time of jobs with machine requirement greater than \(\frac{1}{3}\).
    Note that an optimal schedule would need height at least \(\opt \geq \opt(\calB \sqcup \calM) \geq \frac{P}{2}\) to schedule these jobs since at most two such jobs can be processed concurrently.
    Observe that the schedules \(\sigma_1\) and \(\sigma_2'\) process two jobs at any time except for sparse intervals \(I_1\) and $I_2'$,
    where just one job with machine requirement greater than $\frac13$ is scheduled.
    Therefore:
    \begin{align*}
        h(\sigma_1) + h(\sigma_2')
         & =\frac{P-|I_1|-|I_2'|}{2}+|I_1|+|I_2'|                                   \\
         & =\frac{P}{2} + \frac{|I_1|+|I_2'|}{2} \leq \opt + \frac{|I_1|+|I_2'|}{2}
    \end{align*}
    Note that re-adding the small jobs into the instance does not change the placement of the other jobs, nor the makespan by \cref{obs:sigma1_subset_removal}.
\end{proof}

Now, we can construct \(\sigma_2\) by adding the remaining small jobs \(\calS\) to $\sigma_2'$.
During \Cref{alg:sigma2} we try to place any small job as early as possible.
This might result in jobs being placed below $h(\sigma_2')-|I_2'|$ and definitely results in jobs being placed during $[h(\sigma_2')-|I_2'|,h(\sigma_2'))$
if $h(\sigma_2)>h(\sigma_2')$, as illustrated in \Cref{fig:sigma2_prime,fig:sigma2}.
\begin{algorithm}
    \caption{\(\sigma_2\): Scheduling the Remaining Jobs in \(\calS\)}
    \label{alg:sigma2}
    \begin{algorithmic}[1]
        \Require Intermediate Schedule \(\sigma_2'\).
        \State Place all remaining jobs \(j \in \calS\) via \LS.
        \State Shift any small job ending below \(h(\sigma_2')\) as far up as possible.
    \end{algorithmic}
\end{algorithm}

Note that, if all small jobs are scheduled in \(\sigma_1\) (\(dom(\sigma_1) \cap \calS = \calS\)),
then $h(\sigma_2) = h(\sigma_2')$ and the height can be bounded by \cref{lem:LB opt}.
This leads to the following observation.

\begin{observation}
    \label{obs:assump_sigma_2}
    If \(dom(\sigma_1) \cap \calS = \calS\), then \[h(\sigma_1) + h(\sigma_2) \leq \frac{4}{3}\opt +\pmax.\]
\end{observation}

We consider the remaining case (\(dom(\sigma_1) \cap \calS \neq \calS\)).
The schedule $\sigma_2$ has now two sparse intervals.
The sparse interval which ends below $h(\sigma_2')$ will be denoted as $I_2$ and the other one which starts above $h(\sigma_2')$ as $I_3$.
Note that $|I_3|\leq \pmax$ holds but the height of $I_2$ might be larger.
In the following lemma, we show that the area deficit of \(I_2\) can be compensated by the small jobs if there are enough.
Otherwise, we can bound the total height with \cref{lem:LB opt}.

\begin{lemma}
    \label{lem:ratio_without_tiny}
    If \(dom(\sigma_1) \cap \calS \neq \calS\), then we can construct a schedule \(\sigma\) with:
    \[h(\sigma) \leq \frac{4}{3}\opt + \pmax.\]
\end{lemma}
\begin{proof}
    Let \(\mathcal{S}'\) denote the set of small jobs which start after \(h(\sigma_2') - |I_2'|\) in $\sigma_2$.
    For the analysis, we also consider the fractional part of the (single) job j that crosses $h(\sigma_2')-|I_2'|$.
    This fractional part is then executed during $[h(\sigma_2')-|I_2'|, \sigma_2'(j)+p_j)$.
    Note that all of these jobs have a machine requirement of less than \(\nicefrac{1}{3}\).
    We now distinguish based on the total processing time of these jobs.

    \case{1} \(\sum_{j\in \mathcal{S}'} p_j \le \opt - \frac{|I_1|+|I_2'|}{2}\).
    We construct a new schedule $\sigma$ by scheduling $\sigma_2$ on top of $\sigma_1$.
    First, note that \cref{alg:sigma2} schedules jobs with combined height at least \(|I_2'|\) below \(h(\sigma_2')\)
    (i.e., next to the higher stack in \(\sigma_2'\)).
    Since we schedule at any time (except for $I_2'$ and $I_3$) three small jobs next to each other we increase the height of $h(\sigma_1)+h(\sigma_2')$ by at most
    \((\sum_{j\in \mathcal{S}'} p_j - |I_2'|)/3 + \frac{2}{3}\pmax\) (which is the case if one stack has height $\pmax$ more than the other two).
    \begin{align*}
        h(\sigma) & \leq h(\sigma_1) + h(\sigma_2') + \frac{\sum_{j\in \mathcal{S}'} p_j - |I_2'|}{3} + \frac{2}{3}\pmax  \\
                  & \leq h(\sigma_1) + h(\sigma_2') + \frac{\opt - \frac{|I_1|+|I_2'|}{2} - |I_2'|}{3} + \frac{2}{3}\pmax
    \end{align*}
    With \cref{lem:LB opt}, this implies:
    \begin{align*}
        h(\sigma) & \leq h(\sigma_1) + h(\sigma_2') + \frac{\opt - \frac{|I_1|+|I_2'|}{2} - |I_2'|}{3} +  \frac{2}{3}\pmax                                    \\
                  & \leq \opt + \frac{|I_1|}{2} + \frac{|I_2'|}{2} + \frac{\opt}{3} - \frac{|I_1|}{6} - \frac{|I_2'|}{6} - \frac{|I_2'|}{3}+ \frac{2}{3}\pmax \\
                  & \le \frac{4}{3}\opt+\frac{|I_1|}{3} + \frac{2}{3}\pmax                                                                                    \\
                  & \le \frac{4}{3}\opt+\pmax.
    \end{align*}

    \case{2} \(\sum_{j\in \mathcal{J}'} > \opt -\frac{|I_1|+|I_2'|}{2}\).
    Consider the following construction:
    Split each small job \(j \in \mathcal{S}'\) in two parts.
    One part of width exactly \(\frac{1}{4}\), and a \emph{residue} part of width \(q_j - \frac{1}{4}\) and height $p_j$.
    Note that each small job has a residue part since by definition \(q_j > \frac{1}{4}\) for each \(\mathcal{S}' \subseteq \calS\).
    Denote the set of residue parts by \(\calR\).
    We want to show that these residue parts are big enough to compensate for the area deficit in \(I_2\).
    By the case assumption, we know that:
    \begin{align}
        \label{eq:residue-height}
        \sum_{j\in \calR} p_j & = \sum_{j \in \calJ'}p_j                                                             \\
                              & > \opt - \frac{|I_1|+|I_2'|}{2}  \nonumber                                           \\
                              & \geq |I_2| + |I_1| - \frac{|I_1| + |I_2'|}{2} - \frac{|I_1| + |I_2'|}{2}   \nonumber \\
                              & = |I_2| - |I_2'| \nonumber
    \end{align}

    Where the last inequality holds since \(|I_1| \leq h(\sigma_1)\), \(|I_2| \leq h(\sigma_2')\), and \cref{lem:LB opt} imply $\opt \geq |I_1| + |I_2| - \frac{|I_1| + |I_2'|}{2}$.
    Let \(\tau\) be any time during \(I_2\).
    Since no job in \(\calS'\) has been placed at this time, we know that:
    \(q(\sigma_2, \tau) + \min_{j \in \calS'}q_j > 1\).
    Consequently, we also have: \(q(\sigma_2, \tau) + \min_{j \in \calR}q_j > \frac{3}{4}\) for any time \(\tau\) during \(I_2\).
    So any job $j\in \calS'$ has enough unaccounted residue s.t. we can assume the height of $I_2$ to be reduced by $p_j$.
    With \Cref{eq:residue-height}, this leaves at most height \(|I_2'| \leq \pmax\) that is not covered by the residue parts.
    This leaves an area deficit of at most \((\frac34-\frac23)|I_2'|\leq \frac{1}{12}\pmax\) during $I_2$ since no small job was placed during $I_2$.
    Therefore, we can assume $\sigma_2$ to be as in \cref{def:named_schedules} and apply \cref{obs:missing_area_to_ratio}.
\end{proof}

\subsection{Adding the Tiny Jobs}
\label{sec:tiny-jobs}
In this section, we augment the schedules $\sigma_1^B,\sigma_1^T,\sigma_2$ from the previous section.
We assume that we always have those schedules, but the analysis works just as well if some of these schedules do not schedule a single job.
An operation which operates on such an empty schedule then does nothing instead.
If all schedules are empty, we can execute \LS on all tiny jobs and also have an area deficit of at most $\frac{3}{4}\pmax$.
This holds since at any time below the start time of the last started job we have below $\frac{1}{4}$ unused machines.
Otherwise, we would have started the last job earlier since it has width at most $\frac{1}{4}$.
This means that the schedule has area greater than $\frac{3}{4}h(\sigma)-p_{\max}$.
This gives us our desired ratio, similar to \cref{obs:missing_area_to_ratio}.

\subsubsection{Few Tiny Jobs}
This subsection deals with the situation where we have too few tiny jobs to bound the height by $h(\sigma_1)+h(\sigma_2)\leq \frac43\opt +\pmax$.
Recall that we have previously bound the height of the schedule by $h(\sigma_1)+h(\sigma_2)\leq \frac43\opt+\pmax$ if $dom(\sigma_1)\cap \calS=\calS$ by \cref{obs:assump_sigma_2}.
In this case, we proceed to the next section and skip \cref{alg:tiny_2}.
We assume $dom(\sigma_1)\cap \calS\neq \calS$ in the following.
This assumption enables us to partition $\sigma_1$ into $\sigma_1^B$ and $\sigma_1^T$ by \cref{lem:sigma1_partition}.

Note that before the next algorithm $|I_1^B|$ might be greater than $p_{\max}$ by \cref{lem:sigma1_partition}.
So, we start our next algorithm by greedily scheduling the tiny jobs there.
Afterward, we will reschedule the small jobs in $\sigma_2$ on two stacks, then place as many tiny jobs between the two stacks as possible without changing the makespan.
Lastly, in the following lemma, we collapse the small jobs again such that they are placed as soon as possible.
This procedure is shown in \cref{alg:tiny_2} and illustrated in \cref{fig:alg_tiny_2}.
Notice that if \cref{alg:tiny_2} does not change the makespan, \cref{lem:ratio_without_tiny} gives us our desired result.
So we assume the opposite in this section.
Note that if \cref{alg:tiny_2} changes the height, then $|I_1^B|$ will be smaller than $p_{\max}$.

\begin{algorithm}[ht]
    \caption{Scheduling \(\calT\): Few Tiny Jobs}
    \label{alg:tiny_2}
    \begin{algorithmic}[1]
        \State Perform \LS in $[0,h(\sigma_1^B))$ on $\sigma_1^B$.
        \State Reschedule all jobs in $\sigma_2$ via \LS on two stacks.
        \State Perform \LS in $[0,h(\sigma_2))$ on $\sigma_2$.
    \end{algorithmic}
\end{algorithm}

%\tikzset{font=\fontsize{36}{10}\selectfont\bfseries\boldmath}
%\tikzset{font=\Huge}
\begin{figure}[ht]
    \centering
    \begin{subfigure}{0.20\textwidth}
        \centering
        \begin{tikzpicture}
    %outline of the schedule
    \def\w{10}
    \def\h{5.5}
    \def\xright{\w} % fixed right x-bound
    \def\xleft{0} % fixed left x-bound

    %Sigma 2
    \begin{scope}[fill opacity=0.7]
        %Medium jobs
        \begin{scope}[fill=red]

            % Right stack
            \pgfmathsetmacro{\ycur}{0}
            \expandafter\stackrectright\jobMA
            \expandafter\stackrectright\jobMD
            \stackrectright{3.8}{1}
            \filldraw[fill=blue] (lefttop) ++(0.3, 0) rectangle ++(-2.8, 0.6) node (middle) {};
            \expandafter\stackrectright\jobMF

            \begin{scope}[fill=blue]
                \expandafter\stackrectright\jobSF
            \end{scope}

            % Left stack
            \pgfmathsetmacro{\ycur}{0}
            \expandafter\stackrectleft\jobMB{up}
            \expandafter\stackrectleft\jobMC{up}
            \stackrectleft{4}{0.5}{up}
            \filldraw[fill=blue] (righttop) ++(-0.4, 0.4) rectangle ++(2.5, 0.3);
            \expandafter\stackrectleft\jobME{up}
            \begin{scope}[fill=blue]
                \expandafter\stackrectleft{2.7}{0.7}{up}
                \expandafter\stackrectleft{3}{0.6}{up}

                \draw[dashed, ultra thick] (\w, \ycur) -- (-0.5, \ycur) node[anchor=east,yshift=0.5em] {\(h(\sigma_2)\)};
                % \begin{scope}[color=gapcolor]
                %     %\draw[decorate, decoration={brace, amplitude=10pt}] (\w, \ycur) -- +(0, -0.5) node[midway, xshift=3em] {\(G_3^T\)};
                %     \draw[decorate, decoration={brace, amplitude=10pt}] (0, \ycur-1) -- +(0, 1) node[midway, xshift=-2em, yshift=-0.5em] {\(G_3\)};
                %     % \draw[dashed, thick] (\w, \ycur-0.5) -- (0, \ycur-0.5);
                %     %\draw[decorate, decoration={brace, amplitude=10pt}] (\w, \ycur-0.5) -- +(0, -0.5) node[midway, xshift=3em] {\(G_3^B\)};
                %     \draw[dashed, thick] (\w, \ycur-1) -- (0, \ycur-1);

                %     \draw[dashed, thick] (0, 2) -- +(\w, 0);
                %     \draw[decorate, decoration={brace, amplitude=10pt}] (0, 1) -- +(0, 1) node[midway, xshift=-2em, yshift=-0.5em] {\(G_2\)};
                %     \draw[dashed, thick] (0, 1) -- +(\w, 0);

                % \end{scope}
            \end{scope}

            %middle stack
            \filldraw[fill=blue] (middle) ++(-0.25, 0) rectangle ++(3.1, 0.9);

            \begin{scope}[fill=gray]
                \filldraw (5, 0) rectangle ++(0.5, 0.7)
                ++(0, 0) rectangle ++(-1.3, 0.3)
                ++(-0.1, 0) rectangle ++(0.5, 0.8)
                ++(0, -0.8) rectangle ++(1.3, 1);

            \end{scope}
        \end{scope}

        \draw[very thick] (0, \h) -- (0, 0) -- (\w, 0) -- (\w, \h);
    \end{scope}

\end{tikzpicture}
        \subcaption{\(\sigma_2\) before \ref{alg:tiny_2}}
        \label{fig:alg7_sigma2_before}
    \end{subfigure}
    \begin{subfigure}{0.20\textwidth}
        \centering
        \begin{tikzpicture}
    %outline of the schedule
    \def\w{10}
    \def\h{5.5}
    \def\xright{\w} % fixed right x-bound
    \def\xleft{0} % fixed left x-bound

    %Sigma 2
    \begin{scope}[fill opacity=0.7]
        %Medium jobs
        \begin{scope}[fill=red]

            % Right stack
            \pgfmathsetmacro{\ycur}{0}
            \expandafter\stackrectright\jobMA
            \expandafter\stackrectright\jobMD
            \stackrectright{3.8}{1}
            \expandafter\stackrectright\jobMF

            \begin{scope}[fill=blue]
                \stackrectright{3}{0.6}
                \stackrectright{2.8}{1}
                \expandafter\stackrectright\jobSF
                \draw[dashed, ultra thick] (\w, \ycur) -- (-0.5, \ycur) node[anchor=east,yshift=0.5em] {\(h(\sigma_2)\)};
            \end{scope}

            % Left stack
            \pgfmathsetmacro{\ycur}{0}
            \expandafter\stackrectleft\jobMB{up}
            \expandafter\stackrectleft\jobMC{up}
            \stackrectleft{4}{0.5}{up}
            \expandafter\stackrectleft\jobME{up}
            \begin{scope}[fill=blue]
                \expandafter\stackrectleft{3.1}{0.9}{up}
                \expandafter\stackrectleft{2.7}{0.7}{up}
                \expandafter\stackrectleft{2.5}{0.3}{up}

                % \begin{scope}[color=gapcolor]
                %     %\draw[decorate, decoration={brace, amplitude=10pt}] (\w, \ycur) -- +(0, -0.5) node[midway, xshift=3em] {\(G_3^T\)};
                %     \draw[decorate, decoration={brace, amplitude=10pt}] (0, \ycur-1) -- +(0, 1) node[midway, xshift=-2em, yshift=-0.5em] {\(G_3\)};
                %     % \draw[dashed, thick] (\w, \ycur-0.5) -- (0, \ycur-0.5);
                %     %\draw[decorate, decoration={brace, amplitude=10pt}] (\w, \ycur-0.5) -- +(0, -0.5) node[midway, xshift=3em] {\(G_3^B\)};
                %     \draw[dashed, thick] (\w, \ycur-1) -- (0, \ycur-1);

                %     \draw[dashed, thick] (0, 2) -- +(\w, 0);
                %     \draw[decorate, decoration={brace, amplitude=10pt}] (0, 1) -- +(0, 1) node[midway, xshift=-2em, yshift=-0.5em] {\(G_2\)};
                %     \draw[dashed, thick] (0, 1) -- +(\w, 0);

                % \end{scope}
            \end{scope}

            \begin{scope}[fill=gray]
                \filldraw (5, 0) rectangle ++(0.5, 0.7)
                ++(0, 0) rectangle ++(-1.3, 0.3)
                ++(-0.1, 0) rectangle ++(0.5, 0.8)
                ++(0, -0.8) rectangle ++(1.3, 1)
                ++(0.2, 0) rectangle ++(-2, 0.8)
                ++(0, -1.2) rectangle ++(-0.5, 1)
                 ++(-0.4, 0.2)
                 ++(2, 0.1)   
                 ++(0.9, -0.1) rectangle ++(0.6, 1.5) node (tinytop) {};
        
                \draw[dashed, thick] ($(\w,0)!(tinytop)!(\w, \h)$) -- +(-\w-0.2, 0) node[anchor=east, yshift=-0.5em] {\(C_j\)};
            \end{scope}

        \end{scope}

        \draw[very thick] (0, \h) -- (0, 0) -- (\w, 0) -- (\w, \h);
    \end{scope}

\end{tikzpicture}
        \subcaption{\(\sigma_2\) after \ref{alg:tiny_2}}
        \label{fig:alg7_sigma2_after}
    \end{subfigure}
    \begin{subfigure}{0.29\textwidth}
        \centering
        \begin{tikzpicture}
    %outline of the schedule
    \def\w{10}
    \def\h{5.5}
    \def\xright{\w} % fixed right x-bound
    \def\xleft{0} % fixed left x-bound

    %Sigma 2
    \begin{scope}[fill opacity=0.7]
        %Medium jobs
        \begin{scope}[fill=red]

            % Right stack
            \pgfmathsetmacro{\ycur}{0}
            \expandafter\stackrectright\jobMA
            \expandafter\stackrectright\jobMD
            \stackrectright{3.8}{1}
            \expandafter\stackrectright\jobMF

            \begin{scope}[fill=blue]
                \stackrectright{3}{0.6}
                \stackrectright{2.8}{1}
            \end{scope}

            % Left stack
            \pgfmathsetmacro{\ycur}{0}
            \expandafter\stackrectleft\jobMB{up}
            \expandafter\stackrectleft\jobMC{up}
            \stackrectleft{4}{0.5}{up}
            \expandafter\stackrectleft\jobME{up}
            \begin{scope}[fill=blue]
                \expandafter\stackrectleft{3.1}{0.9}{up}

                \expandafter\stackrectleft{2.7}{0.7}{up}
                \draw[dashed, thick] (0.8, \ycur) -- (\w+0.5, \ycur) node[anchor=west] {\(h(\sigma_2)-|I_3^T|\)};

                \draw[dashed, ultra thick] (-0.5, \ycur) node[anchor=east,yshift=0.5em, color=white] {\(h(\sigma_2)\)};
                % \begin{scope}[color=gapcolor]
                %     %\draw[decorate, decoration={brace, amplitude=10pt}] (\w, \ycur) -- +(0, -0.5) node[midway, xshift=3em] {\(G_3^T\)};
                %     \draw[decorate, decoration={brace, amplitude=10pt}] (0, \ycur-1) -- +(0, 1) node[midway, xshift=-2em, yshift=-0.5em] {\(G_3\)};
                %     % \draw[dashed, thick] (\w, \ycur-0.5) -- (0, \ycur-0.5);
                %     %\draw[decorate, decoration={brace, amplitude=10pt}] (\w, \ycur-0.5) -- +(0, -0.5) node[midway, xshift=3em] {\(G_3^B\)};
                %     \draw[dashed, thick] (\w, \ycur-1) -- (0, \ycur-1);

                %     \draw[dashed, thick] (0, 2) -- +(\w, 0);
                %     \draw[decorate, decoration={brace, amplitude=10pt}] (0, 1) -- +(0, 1) node[midway, xshift=-2em, yshift=-0.5em] {\(G_2\)};
                %     \draw[dashed, thick] (0, 1) -- +(\w, 0);

                % \end{scope}
            \end{scope}

            \begin{scope}[fill=gray]
                \filldraw (5, 0) rectangle ++(0.5, 0.7)
                ++(0, 0) rectangle ++(-1.3, 0.3)
                ++(-0.1, 0) rectangle ++(0.5, 0.8)
                ++(0, -0.8) rectangle ++(1.3, 1)
                ++(0.2, 0) rectangle ++(-2, 0.8)
                ++(0, -1.2) rectangle ++(-0.5, 1)
                ++(2.5, 0.2) rectangle ++(0.6, 1.5) node (tinytop) {};

                \filldraw[fill=blue] (tinytop) ++(-0.8, -1.5) rectangle ++(-2.5, 0.5)
                ++(0, 0) rectangle ++(2.5, 0.3) node (smalltop) {};

                \draw[dashed, thick] ($(\w,0)!(tinytop)!(\w, \h)$) -- +(-\w-0.4, 0) node[anchor=east, yshift=-0.5em] {\(C_j\)};
                %\draw[dashed, thick] ($(0,0)!(smalltop)!(0, \h)$) -- +(\w+0.2, 0) node[anchor=west, yshift=-0.5em] {\(h(\sigma_2)-|I_3^T|\)};
            \end{scope}

        \end{scope}

        \draw[very thick] (0, \h) -- (0, 0) -- (\w, 0) -- (\w, \h);
    \end{scope}

\end{tikzpicture}
        \subcaption{\(\sigma_2\) as in \ref{lem:tiny_2}}
        \label{fig:alg7_proof}
    \end{subfigure}
    \caption{Illustrations around \cref{alg:tiny_2}}
    \Description{
        In this figure we see three schedules which depend on how many tiny jobs we have.
        In the first schedule we have barely enough tiny jobs to fill the gaps, which means that the makespan does not change, which gives us our desired ratio by the previous section.
        In the second and third schedule we have enough tiny jobs to fill the gaps and therefore we need to change the schedule.
        The second schedule is an intermediate schedule which is created by scheduling all medium and small jobs on two stacks and scheduling the tiny jobs greedily between.
        In the third schedule, we reschedule then the small jobs on three stacks so that the height difference is at most $p_{max}$.
        In this case the schedule has a tiny job which ends between $h(\sigma_2)-|I_3^T|$ and $h(\sigma_2)$.
    }
    \label{fig:alg_tiny_2}
\end{figure}

\begin{lemma}
    \label{lem:tiny_2}
    Given the schedules \(\sigma_1^B\), \(\sigma_1^T\), and \(\sigma_2\), as constructed by \cref{alg:tiny_2},
    we can construct a schedule $\sigma$ with \(h(\sigma)\leq \frac{4}{3}\opt + \pmax\).
\end{lemma}

\begin{proof}
    To build $\sigma$, we first alter $\sigma_2$ by removing all small jobs which complete after the last tiny job has completed and remove them.
    Then, we reschedule those jobs by \LS on top.
    Let $j$ be the tiny job with the highest completion time.
    In addition, we partition $I_3$ into two subintervals $I_3^T$ and $I_3^B$, such that the upper subinterval $I_3^T$ has a machine utilization of at most $\frac{1}{3}$.
    We consider two cases, depending on $C_j$ the completion time of $j$.

    \case{1} $C_j< h(\sigma_2)- |I_3^T|$.
    This means that $\sigma_2$ has again two sparse intervals as specified in \cref{def:named_schedules}.
    Furthermore, $\sigma_1^T$ does not change and, as mentioned earlier, in $\sigma_1^B$ the height of $I_1^B$ is again below $p_{\max}$
    since we assumed that the makespan would change.
    So, we can use \cref{lem:combine_schedules} to prove the makespan.

    \case{2} $C_j \geq h(\sigma_2)- |I_3^T|$.
    Let $I_3'$ be the sparse interval in $\sigma_2$ with width at most $\frac{2}{3}$ and $I_2'$ the sparse interval of $\sigma_2$ with width greater than $\frac23$.
    Note that $|I_2'|,|I_3'|\leq \pmax$ since $I_3'$ is a subinterval of $I_3$ and $I_2'$ is a subinterval of $[\sigma_2(j),\sigma_2(j)+p_j)$ see \Cref{fig:alg7_proof}.
    By using the same technique from \cref{lem:combine_schedules}, we create a new schedule by
    placing $\sigma_2$ at $0$, $\sigma_1^T$ at $h(\sigma_2)-\min(|I_1^T|, |I_3'|)$ and $\sigma_1^B$ on top of that.
    The area deficit is now at most $ (\frac34-\frac23)|I_2'|+(\frac34-\frac14)|I_3'|+(\frac34-\frac23)|I_1^B|$ or at most $(\frac34-\frac23)|I_2'|+(\frac34-\frac12)|I_1^T|+(\frac34-\frac23)|I_1^B|$.
    Since the area deficit of $\sigma$ is in both cases at most $\frac{3}{4}\pmax$, the lemma follows from \cref{obs:missing_area_to_ratio}.
    Note that even if $j$ completes after every non-tiny job, the area deficit during $[\sigma_2(j),\sigma_2(j)+p_j)$ is at most $(\frac34-\frac14)p_j$ since in \cref{alg:tiny_2} we scheduled the tiny jobs to have a completion time which is at most the maximal completion time of the non-tiny jobs.
    Therefore, during $I_2$ we have enough area.
\end{proof}

In the following we show how to extend \(\sigma_1\) and \(\sigma_2\) if \cref{alg:tiny_2} terminates without placing all tiny jobs.

\begin{observation}
    \label{obs:alt_schedule_tiny2}
    If \cref{alg:tiny_2} terminates without placing all tiny jobs and outputs $\sigma_1$ and $\sigma_2$.
    Then $h(\sigma_1)+h(\sigma_2)\leq\frac{4}{3}\opt + \pmax$.
\end{observation}
\begin{proof}
    In the following, we consider two cases and show in both that the area deficit of $\sigma_2$ on top of $\sigma_1$ can be bound by $\frac{3}{4}\pmax$.
    Therefore, with \cref{obs:missing_area_to_ratio}, we have shown the claim.
    Note that after \cref{alg:tiny_2} in $\sigma_2$ during $(h(\sigma_2)-p_{\max}, h(\sigma_2)]$ we only have a single sparse interval.
    We define it as $I_2$ (merging the previous $I_2$ and $I_3$).

    \case{1} $|I_1^B|+|I_1^T| < \pmax$.
    Note that the area deficit is bounded by $\deficit(\sigma)\leq \frac{1}{4}|I_1^B|+\frac{1}{12}|I_1^T|+\frac{1}{4}|I_2^B|+\frac{1}{2}|I_2^T|$.
    With $|I_2^B|+|I_2^T|< \pmax$ the observation follows immediately.

    \case{2} $|I_1^B|+|I_1^T|\geq \pmax$.
    We can write $I_2=I_2^B\sqcup I_2^T$ where during $I_2^T$ one non-tiny job is being scheduled and during $I_2^B$ two non-tiny jobs are being scheduled.
    Let $j$ be a non-tiny job with the largest completion time in $\sigma_2$.
    Note that during $I_2^T$, at least $q_j$ machines are used and $|I_2^T|\leq p_j$.
    Consider $I_1=I_1^B\sqcup I_2^T$ and $\bar{I_1}\coloneq [{i_1}_{start},{i_1}_{end}-\pmax]$ where $I_1=[{i_1}_{start},{i_1}_{end})$.
    Since $j$ was not scheduled during $\bar{I_1}$ even though there was enough available height, we know that the combined machine usage ($q_j$ and $q(\sigma,t),t\in \bar{I_1}$) is greater than $1$.
    Also, we only assumed the width during $I_2^T$ to be greater than $\frac14$ and the width during $\bar{I_1}\subseteq I_1^B$ to be greater than $\frac23$.
    Therefore, we can assume $\frac{1}{12}=(1-\frac{2}{3}-\frac14)$ additional machine usage either during $\bar{I_1}$ or during $I_2^T$ (depending on which is smaller).
    Specifically, if $|\bar{I_1}|>|I_2^T|$ then $\bar{I_1}$ can be assumed to be wider by $\frac{1}{12}$, so during $\bar{I_1}$ there is no area deficit.
    The area deficit is then at most $(\frac34-\frac{1}{2})(|I_1|-|\bar{I_1}|)+(\frac34-\frac{1}{4})|I_2|\leq \frac34\pmax$.
    Otherwise, $I_2^T$ can be assumed to be $\frac{1}{12}$ wider; therefore the area deficit is at most $\frac{1}{12}|I_1^B|+\frac{1}{4}|I_1^T|+(\frac34+\frac{1}{12}-\frac{1}{4})|I_2^T|+\frac14|I_2^B|\leq \frac34 \pmax$.
\end{proof}

\subsubsection{Many Tiny Jobs}
\label{sec:many_tiny_jobs}
Here we assume \cref{alg:tiny_2} terminates without placing all jobs (or we skipped the algorithm).
In order to schedule these jobs we have to further transform the schedule as illustrated in \cref{fig:many}.
For this we first remove all tiny jobs, $\sigma_2$ consists now of two stacks.
Then remove all non-big jobs from $\sigma_1$ and all jobs from the higher stack in $\sigma_2$.
Afterward we replace this stack by executing \LS on a single stack on the removed non-tiny jobs.
The tiny jobs will be added in the next algorithm.
The advantage of this construction is that $\sigma_1$ and $\sigma_2$ now have a simpler structure, while we can still interleave them so that the height remains unchanged (see \cref{fig:many}).

\begin{algorithm}[ht]
    \caption{Transformation for Many Tiny Jobs}
    \label{alg:tiny_3}
    \begin{algorithmic}[1]
        \Require Schedules \(\sigma_1\), \(\sigma_2\) from \cref{alg:tiny_2} or \cref{lem:LB opt}
        \State Remove the tiny jobs.
        \State Remove all non-big jobs $\hat\calJ \coloneqq dom(\sigma_1) \setminus \calB$ in $\sigma_1$.
        \State Sort the jobs in $\hat\calJ$ into the higher stack in $\sigma_2$.
    \end{algorithmic}
\end{algorithm}

Note that the resulting schedules \(\sigma_1\) and \(\sigma_2\), which were constructed by \cref{alg:tiny_3}, both have a single sparse interval.
We call them now \(I_1\) and \(I_2\), respectively.
Both intervals may be larger than \(\pmax\).

In one last algorithmic component we then place the tiny jobs via \LS, where we make sure to place each job in the larger interval at each step (\cref{alg:tiny_4}).

\begin{algorithm}[ht]
    \caption{Scheduling \(\calT\): Many Tiny Jobs}
    \label{alg:tiny_4}
    \begin{algorithmic}[1]
        \Require Schedules \(\sigma_1\), \(\sigma_2\) from \cref{alg:tiny_3}
        \For{\(j \in \calT\)}
        \State Place \(j\) with \LS in the higher interval.
        \EndFor
    \end{algorithmic}
\end{algorithm}

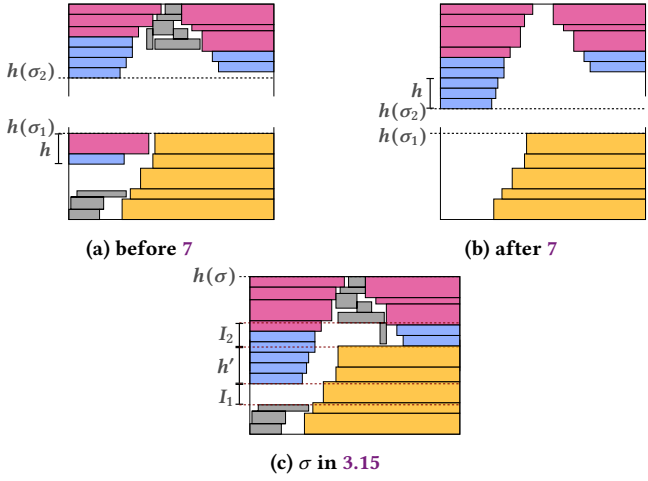
\begin{figure}[ht]
    \centering
    \begin{subfigure}{0.20\textwidth}
        \centering
        \begin{tikzpicture}

    %outline of the schedule

    \def\w{10}

    \def\h{4.5}

    \def\xright{\w} % fixed right x-bound

    \def\xleft{0} % fixed left x-bound

    %Sigma 2

    \begin{scope}[fill opacity=0.7, shift={(0, \h+1.3)}, yscale=-1]

        %Medium jobs

        \begin{scope}[fill=red]

            \pgfmathsetmacro{\ycur}{0} % Initial y value

            % Draw stacked rectangles

            \expandafter\stackrectright\jobMA

            \expandafter\stackrectright\jobMD

            \pgfmathsetmacro{\middlestart}{\ycur} % Adjust y value for next job

            \expandafter\stackrectright\jobME

            \pgfmathsetmacro{\rightstack}{\ycur}

            %Small jobs

            \begin{scope}[fill=blue]

                \expandafter\stackrectright\jobSC

                \expandafter\stackrectright\jobSE

            \end{scope}

            \pgfmathsetmacro{\ycur}{0}

            \expandafter\stackrectleft\jobMB{up}

            \expandafter\stackrectleft\jobMC{up}

            \expandafter\stackrectleft\jobMF{up}

            \begin{scope}[fill=blue]

                \expandafter\stackrectleft\jobSA{up}

                \expandafter\stackrectleft\jobSB{up}

                \expandafter\stackrectleft\jobSD{up}

                \expandafter\stackrectleft\jobSF{up}

                \draw[dashed, ultra thick] (-0.5, \ycur) node[anchor=east, yshift=0.9em] {\(h(\sigma_2)\)} -- (\w, \ycur);

                %Tiny jobs

                \begin{scope}[fill=gray]

                    \filldraw (4.7, 0) rectangle ++(0.8, 0.5)

                    rectangle ++(-1.2, 0.3)

                    ++(-0.2, 0) rectangle ++(1, 0.7)

                    ++(-0, -0.3) rectangle ++(0.7, 0.5)

                    ++(0.6, 0) rectangle ++(-2.2, 0.5)

                    ++(-0.4, -1) rectangle ++(0.3, 1);

                \end{scope}

                % \begin{scope}[color=gapcolor]

                %     \draw[dashed, thick] (-0.5, 1) -- +(\w+0.5, 0);

                %     \draw[decorate, decoration={brace, amplitude=10pt}] (0, \ycur) -- (0, 1) node[midway, xshift=-2.5em, yshift=-0.5em] {\(G_2\)};

                % \end{scope}

            \end{scope}

            % %Tiny jobs

            % \begin{scope}[fill=gray!50]

            %     \fill (4.7, 0) rectangle ++(0.8, 0.5)

            %     rectangle ++(-1.2, 0.3);

            % \end{scope}

        \end{scope}

        \draw[very thick] (0, \h) -- (0, 0) -- (\w, 0) -- (\w, \h);

    \end{scope}

    %Sigma 1

    %\begin{scope}[shift={(0, -\h -1)}, fill opacity=0.7]

    \begin{scope}[shift={(0, -\h -0.2)}, fill opacity=0.7]

        %Tiny jobs

        \begin{scope}[fill=gray]

            \filldraw (0, 0) rectangle ++(1.5, 0.5)

            -- ++(0.2, 0) rectangle ++(-1.6, 0.6)

            -- ++(0.3, 0) rectangle ++(2.4, 0.3);

        \end{scope}

        %Big jobs

        \begin{scope}[fill=yellow]

            \pgfmathsetmacro{\ycur}{0} % Initial y value

            % Draw stacked rectangles

            \expandafter\stackrectright\jobBA

            \expandafter\stackrectright\jobBB

            \pgfmathsetmacro{\cutheight}{\ycur}

            \expandafter\stackrectright\jobBC

            \expandafter\stackrectright\jobBD

            \expandafter\stackrectright\jobBE

            %Medium jobs

            \begin{scope}[fill=red]

                %\expandafter\stackrectleft\jobMA{down}

                \expandafter\stackrectleft\jobMG{down}

                %Small jobs

                \begin{scope}[fill=blue]

                    \expandafter\stackrectleft\jobSG{down}

                \end{scope}

            \end{scope}

            \draw[|-|, ultra thick] (-0.5, \ycur) -- +(0, -1.5) node[midway, xshift=-2em] {\(h\)};

            \draw[dashed, ultra thick] (-0.5, \ycur) node[anchor=east, yshift=0.9em] {\(h(\sigma_1)\)} -- (\w, \ycur);

            % \begin{scope}[color=gapcolor]

            %     \draw[dashed, thick] ($(leftbottom)+(-0.5, 0)$) -- +(\w+0.5, 0);

            %     \draw[decorate, decoration={brace, amplitude=10pt}] (0, 0) -- (leftbottom) node[midway, xshift=-2.5em, yshift=-0.5em] {\(G_1\)};

            % \end{scope}

        \end{scope}

        \draw[very thick] (0, \h) -- (0, 0) -- (\w, 0) -- (\w, \h);

    \end{scope}

\end{tikzpicture}
        \subcaption{before \ref{alg:tiny_3}}
        \label{fig:many_before_transformation}
    \end{subfigure}
    \hfill
    \begin{subfigure}{0.20\textwidth}
        \centering
        \begin{tikzpicture}

    %outline of the schedule
    \def\w{10}
    \def\h{4.5}
    \def\xright{\w} % fixed right x-bound
    \def\xleft{0} % fixed left x-bound

    %Sigma 2
    \begin{scope}[shift={(0, \h+3)}, yscale=-1, fill opacity=0.7]
        %Medium jobs
        \begin{scope}[fill=red]

            \pgfmathsetmacro{\ycur}{0} % Initial y value
            % Draw stacked rectangles
            \expandafter\stackrectright\jobMA
            \expandafter\stackrectright\jobMD
            \pgfmathsetmacro{\middlestart}{\ycur} % Adjust y value for next job
            \expandafter\stackrectright\jobME
            \pgfmathsetmacro{\rightstack}{\ycur}
            %Small jobs
            \begin{scope}[fill=blue]
                \expandafter\stackrectright\jobSC
                \expandafter\stackrectright\jobSE
            \end{scope}
            \pgfmathsetmacro{\ycur}{0}
            \expandafter\stackrectleft\jobMB{up}
            \expandafter\stackrectleft\jobMC{up}
            \expandafter\stackrectleft\jobMG{up}
            \expandafter\stackrectleft\jobMF{up}
            \begin{scope}[fill=blue]
                \expandafter\stackrectleft\jobSA{up}
                \expandafter\stackrectleft\jobSB{up}
                \expandafter\stackrectleft\jobSD{up}
                \expandafter\stackrectleft\jobSG{up}
                \expandafter\stackrectleft\jobSF{up}

                \draw[|-|, ultra thick] (-0.5, \ycur) -- +(0, -1.5) node[midway, yshift=0.5em, xshift=-2em] {\(h\)};

                \draw[dashed, ultra thick] (-0.5, \ycur) node[anchor=east, yshift=-0.5em] {\(h(\sigma_2)\)} -- (\w, \ycur);

                % \begin{scope}[color=gapcolor]
                %     \draw[dashed, thick] (-0.5, 1) -- +(\w+0.5, 0);
                %     \draw[decorate, decoration={brace, amplitude=10pt}] (0, \ycur) -- (0, 1) node[midway, xshift=-2.5em, yshift=-0.5em] {\(G_2\)};

                % \end{scope}

            \end{scope}

            % %Tiny jobs
            % \begin{scope}[fill=gray!50]
            %     \fill (4.7, 0) rectangle ++(0.8, 0.5)
            %     rectangle ++(-1.2, 0.3);

            % \end{scope}

        \end{scope}
        \draw[very thick] (0, \h) -- (0, 0) -- (\w, 0) -- (\w, \h);
    \end{scope}

    %Sigma 1
    %\begin{scope}[shift={(0, -\h - 1)}, fill opacity=0.7]
    \begin{scope}[shift={(0, -\h +1.5)}, fill opacity=0.7]

        % %Tiny jobs
        % \begin{scope}[fill=gray!50]
        %     \fill (0, 0) rectangle ++(1.5, 0.5)
        %     -- ++(0.2, 0) rectangle ++(-1.2, 0.5)
        %     -- ++(-0.3, 0) rectangle ++(2.4, 0.1);
        % \end{scope}

        %Big jobs
        \begin{scope}[fill=yellow]

            \pgfmathsetmacro{\ycur}{0} % Initial y value
            % Draw stacked rectangles
            \expandafter\stackrectright\jobBA
            \expandafter\stackrectright\jobBB
            \pgfmathsetmacro{\cutheight}{\ycur}
            \expandafter\stackrectright\jobBC
            \expandafter\stackrectright\jobBD
            \expandafter\stackrectright\jobBE

            \draw[dashed, ultra thick] (-0.5, \ycur) node[anchor=east, yshift=-0.5em] {\(h(\sigma_1)\)} -- (\w, \ycur);
            % \begin{scope}[color=gapcolor]
            %     \draw[decorate, decoration={brace, amplitude=10pt}] (0, 0) -- (0, \ycur) node[midway, xshift=-2.5em, yshift=-0.5em] {\(G_1\)};

            % \end{scope}
        \end{scope}

        \draw[very thick] (0, \h) -- (0, 0) -- (\w, 0) -- (\w, \h);
    \end{scope}

\end{tikzpicture}
        \subcaption{after \ref{alg:tiny_3}}
        \label{fig:many_after_transformation}
    \end{subfigure}
    \hfill
    \begin{subfigure}{0.20\textwidth}
        \centering
        \begin{tikzpicture}

    %outline of the schedule
    \def\w{10}
    \def\h{4.5}
    \def\xright{\w} % fixed right x-bound
    \def\xleft{0} % fixed left x-bound

    %Sigma 2
    \begin{scope}[yscale=-1, shift={(0, -\h)}, fill opacity=0.7]

        %Medium jobs
        \begin{scope}[fill=red]

            \pgfmathsetmacro{\ycur}{0} % Initial y value
            % Draw stacked rectangles
            \expandafter\stackrectright\jobMA
            \expandafter\stackrectright\jobMD
            \pgfmathsetmacro{\middlestart}{\ycur} % Adjust y value for next job
            \expandafter\stackrectright\jobME
            \pgfmathsetmacro{\rightstack}{\ycur}
            %Small jobs
            \begin{scope}[fill=blue]
                \expandafter\stackrectright\jobSC
                \expandafter\stackrectright\jobSE
            \end{scope}
            \pgfmathsetmacro{\ycur}{0}
            \expandafter\stackrectleft\jobMB{up}
            \expandafter\stackrectleft\jobMC{up}
            \expandafter\stackrectleft\jobMG{up}
            \expandafter\stackrectleft\jobMF{up}
            \begin{scope}[fill=blue]
                \expandafter\stackrectleft\jobSA{up}
                \expandafter\stackrectleft\jobSB{up}
                \expandafter\stackrectleft\jobSD{up}
                \expandafter\stackrectleft\jobSG{up}
                \expandafter\stackrectleft\jobSF{up}

                %\draw[|-|] (-0.5, \ycur) -- +(0, -1.5) node[midway, xshift=-1.5em, thick] {\(h\)};

                \draw[|-|, ultra thick] (-0.5, \ycur) -- +(0, -1.75) node[midway, xshift=-2em, thick] {\(h'\)};

            \end{scope}

            %Tiny jobs
            \begin{scope}[fill=gray]
                \filldraw (4.7, 0) rectangle ++(0.8, 0.5)
                rectangle ++(-1.2, 0.3)
                ++(-0.2, 0) rectangle ++(1, 0.7)
                ++(-0, -0.3) rectangle ++(0.7, 0.5)
                ++(0.6, 0) rectangle ++(-2.2, 0.5)
                ++(2, 0) rectangle ++(0.3, 1);
            \end{scope}

            \begin{scope}[color=gapcolor]
                \draw[dashed, thick] (-0.5, 2.2) -- +(\w+0.5, 0);
                \draw[|-|, ultra thick, color=black] (-0.5, 3.35) -- (-0.5, 2.2) node[midway, xshift=-2em, yshift=-0.0em] {\(I_2\)};
                \draw[dashed, thick] (-0.5, 3.35) -- +(\w+0.5, 0);
            %\draw[|-|, ultra thick, color=black] (-0.5, \ycur) -- (-0.5, \righttop) node[midway, xshift=-2em, yshift=-0.0em] {\(I_2'\)};

            \end{scope}

        \end{scope}
        \draw[very thick] (0, \h) -- (0, 0) -- (\w, 0) -- (\w, \h);
        \draw[dashed, very thick] (-0.5, 0) node[anchor=east, yshift=0.0em] {\(h(\sigma)\)}-- +(\w+0.5, 0);
    \end{scope}

    %Sigma 1
    %\begin{scope}[shift={(0, -\h - 1)}, fill opacity=0.7]
    \begin{scope}[shift={(0, -\h +1.5)}, fill opacity=0.7]

        %Tiny jobs
        \begin{scope}[fill=gray]
            \filldraw (0, 0) rectangle ++(1.5, 0.5)
            -- ++(0.2, 0) rectangle ++(-1.6, 0.6)
            -- ++(0.3, 0) rectangle ++(2.4, 0.3);
        \end{scope}

        %Big jobs
        \begin{scope}[fill=yellow]

            \pgfmathsetmacro{\ycur}{0} % Initial y value
            % Draw stacked rectangles
            \expandafter\stackrectright\jobBA
            \expandafter\stackrectright\jobBB
            \pgfmathsetmacro{\cutheight}{\ycur}
            \expandafter\stackrectright\jobBC
            \expandafter\stackrectright\jobBD
            \expandafter\stackrectright\jobBE

            \begin{scope}[color=gapcolor]
                \draw[dashed, thick] (-0.5, 2.4) -- +(\w+0.5, 0);
                \draw[|-|, ultra thick, color=black] (-0.5, 1.4) -- (-0.5, 2.4) node[midway, xshift=-2em, yshift=-0.5em] {\(I_1\)};
                \draw[dashed, thick] (-0.5, 1.4) -- +(\w+0.5, 0);

            \end{scope}
        \end{scope}

        \draw[very thick] (0, \h) -- (0, 0) -- (\w, 0) -- (\w, \h);
    \end{scope}

\end{tikzpicture}
        \subcaption{$\sigma$ in \ref{lem:many_tiny}}
        \label{fig:many_sigma}
    \end{subfigure}
    \caption{$\sigma_1$ and $\sigma_2$ (mirrored)}
    \label{fig:many}
    \Description{
        Depicted are three figures.
        In all cases we see the two schedules $\sigma_1$ and $\sigma_2$.
        For a better visualization we see $\sigma_2$ upside down above $\sigma_1$, so that the higher stack of $\sigma_1$ is directly above the medium and small jobs of $\sigma_2$.
        In the second figure me move the medium and small jobs from $\sigma_2$ up and sort them into the higher stack of $\sigma_1$.
        Since this might also affect the placement of the tiny jobs, we removed them in this picture.
        In the third figure, we move $\sigma_2$ down as much as possible and therefore by at least $h$ plus the height of the small and medium jobs which were sorted into the higher stack.
        Also the tiny jobs are again in the visualization.
        It might be the case that we have used all tiny jobs after this transformation, but since we did not change the makespan it is valid.
    }
\end{figure}

\begin{lemma}
    \label{lem:many_tiny}
    Given $\sigma_1, \sigma_2$ as constructed by \cref{alg:tiny_4}, we can construct a schedule $\sigma$ with $h(\sigma)\leq \frac{4}{3}\opt +\pmax$.
\end{lemma}
\begin{proof}
    The proof is split in two parts.
    Either we do not change the height of the resulting schedule, then we have the same height as in \cref{obs:alt_schedule_tiny2} (or \cref{obs:assump_sigma_2}).
    Or the height does change, then we can bound the area deficit to \(\frac{3}{4}\pmax\).

    Construct the new schedule $\sigma$ by placing $\sigma_1$ at $0$ and (rotated) $\sigma_2$ at $h(\sigma_1)-h'$ such that \(h'\in \mathbb{Q}\) is maximal.
    The resulting schedule is illustrated in \Cref{fig:many_sigma}.
    Note that in $\sigma$ the height of both sparse intervals \(I_1\) and \(I_2\) is reduced by exactly \(h'\).
    Denote now the height of all non-big jobs that were removed from \(\sigma_1\) during \cref{alg:tiny_3} as
    \(h \coloneqq \sum_{j \in \hat{\calJ}}{p_j}\) where \(\hat{\calJ}\) is defined as in \cref{alg:tiny_3}.
    Assume w.l.o.g. that \(I_1\) is the larger interval, i.e., \(|I_1|\geq |I_2|\). We can now distinguish two cases.

    \case{1} \(|I_1| > \pmax\).
    In this case in $\sigma_1$ all tiny jobs are completed below $h(\sigma_1)-h$.
    Otherwise, the interval would be smaller, since \LS with tiny jobs creates sparse intervals with height at most $p_{\max}$.
    Therefore, we can choose $h'\geq h$.
    Which implies that, the height of $\sigma$ is at most $h(\sigma_1)=h(\sigma_2)-h$ which is just the height of \cref{obs:alt_schedule_tiny2}.

    \case{2} \(|I_1| \leq \pmax\).
    Note that the machine usage in \(I_1\) and \(I_2\) is non-increasing by \cref{def:named_schedules} since we use \LS.
    Consider now the width of the two sparse intervals $I_1$ and $I_2$ after combining the schedules.
    We know that $q_{min}(I_1)+q_{min}(I_2)>1$ otherwise we could choose a larger $h'$ which is a contradiction to the maximality of $h'$.
    Therefore, the area in the interval is at least
    \begin{align*}
        q_{min}(I_1)|I_1|+q_{min}(I_2)|I_2|\geq q_{min}(I_1)|I_2|+q_{min}(I_2)|I_2|>1\cdot |I_2|.
    \end{align*}
    So the area deficit is at most $\frac{3}{4}|I_1|\leq \frac{3}{4}\pmax$.
    With \cref{obs:missing_area_to_ratio} we have proven the claim.
\end{proof}
Since \cref{alg:tiny_4} only terminates when all tiny jobs are placed, we have proven \cref{thm:pts}.

\section{Extension to \MCS}
In this section, we show how to extend our algorithm to \MCS (MCS), where we have $N > 1$ disjoint clusters of identical width.
We aim to minimize the global makespan, i.e., the maximum completion time of any job across all clusters.
Note that there exists a unique $a\in\{0,1,2\}$ such that $N=3i+1+a$ for some $i\in \mathbb{Z}_{\geq 0}$.
Given an instance, let $\opt$ be the optimal makespan of the MCS problem and let $\opt_{\text{PTS}}$ be the optimal makespan of the corresponding PTS-setting,
i.e., when scheduling all jobs on a single cluster.

To prove our first MCS result, we use the partitioning technique that was originally proposed by Jansen and Rau~\cite{JansenR19} (\cref{alg:mcs_asymp}).
The improved approximation ratio of our PTS algorithm (\cref{sec:pts}) directly improves their result, as stated in their work.
We state the algorithm here tailored to our PTS algorithm for completeness.
\begin{algorithm}
    \caption{MCS1}
    \label{alg:mcs_asymp}
    \begin{algorithmic}[1]
        \State Run the PTS algorithm (\cref{sec:pts}) with approximation ratio \(\frac{4}{3}\opt_{\text{PTS}} +\pmax\) on a single cluster.
        \State Cut the schedule into \(2i+1\) pieces of equal height with \(2i\) cuts.
        \State Remove the jobs processed during the cuts.
        \State Assign every piece its own cluster.
        \State Schedule the removed jobs from two cuts on a new cluster.
    \end{algorithmic}
\end{algorithm}

\mcsAsymp*
\begin{proof}
    Let $h$ be the height of the schedule returned by our PTS algorithm.
    Then the makespan $T$ over all clusters with a single piece can be bounded by:
    \begin{align*}
        \frac{h}{2i+1}
         & \le \left(\frac{4}{3}\opt_{\text{PTS}} + \pmax\right) / \left(2i + 1\right)                 \\
         & \leq \left(\frac{4}{3}\opt\cdot N + \pmax\right) / \left(2i + 1\right)                      \\
         & = \frac{12i + 4+4a}{6i + 3}\opt + \frac{3}{6i + 3}\pmax                                     \\
         & \leq \frac{12i+7+4a}{6i+3}\opt                                              & \pmax\leq\opt
    \end{align*}
    The first inequality holds due to $\cref{thm:pts}$ and the second inequality holds because $\opt \geq \frac{\opt_{\text{PTS}}}{N}$.
    The final equation holds because we can write $N$ uniquely as $N=3i+1+a$ for $i\in\mathbb{Z}_{\geq 0}$ and $a\in \{0,1,2\}$.
    If now $N \rightarrow \infty$, we have $i\rightarrow\infty$ and therefore, $T \overset{N\rightarrow\infty}{\le} 2\opt$.
    Finally, every cluster where two cuts have been placed has makespan of at most \(2\pmax\leq 2\opt\).
    Therefore, we only have to consider the clusters to which the pieces have been assigned.
\end{proof}

\newcolumntype{C}{>{\centering\arraybackslash}X}
\begin{table}[htbp]
    \centering
    \renewcommand{\arraystretch}{1.5}
    \begin{tabularx}{\columnwidth}{| >{\bfseries}l | C | C | C | C |}
        \hline
        N     & 4                                     & 7                                       & 10                           & 13                           \\ \hline
        Ratio & $\frac{19}{9} \approx 2.\overline{1}$ & $\frac{31}{15} \approx 2.0\overline{6}$ & $\frac{43}{21} \approx 2.05$ & $\frac{55}{27} \approx 2.04$ \\ \hline
    \end{tabularx}
    \caption{Convergence of \cref{alg:mcs_asymp} if $N = 3i+1$}
    \vspace{-7mm}
\end{table}

Note that the approximation ratio converges to \(2\) but has a worst-case ratio of \(3\) for $i\geq 1$.
Motivated by this, our next step is to modify the algorithm to have a better worst-case bound for a small number of clusters.

\begin{algorithm}
    \caption{MCS2}
    \label{alg:mcs_absolute}
    \begin{algorithmic}[1]
        \State \textbf{If \(N = 3i\) then} Use the algorithm by Jansen and Rau~\cite{JansenR19}.
        \State \textbf{else if \(N = 3i+1\) then} Use MCS1.
        \State \textbf{else if \(N=3i+2\) then}
        \State Consider the PTS schedule $\sigma$.
        \State Place all jobs which were completed after $h(\sigma)-\pmax$ on the additional cluster.
        \State Use MCS1 with \(N \coloneqq 3i+1\).
        \State \textbf{end if}
    \end{algorithmic}
\end{algorithm}

\begin{figure*}[t]
    \centering
    \includegraphics[width=0.8\textwidth]{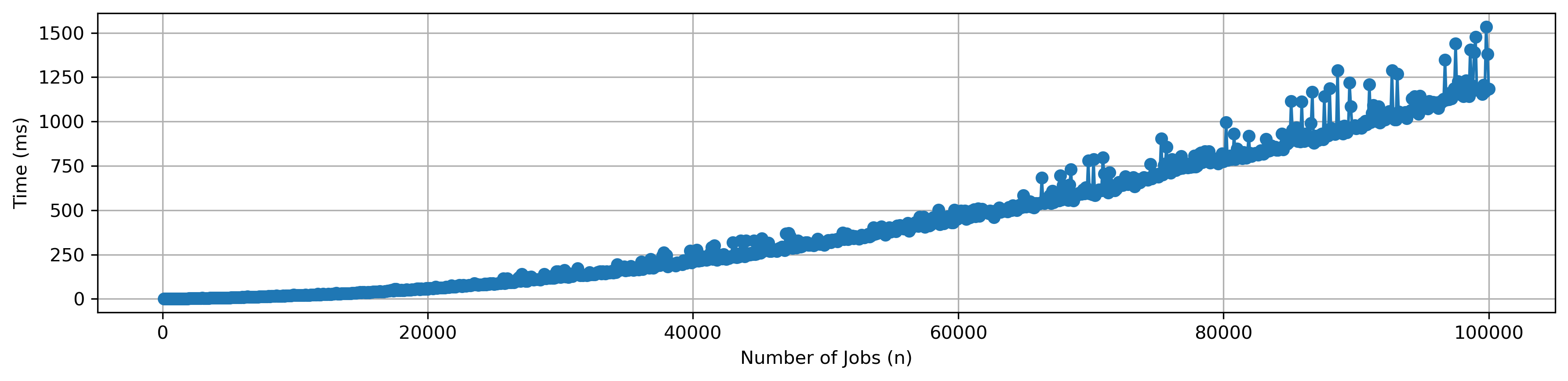}
    \caption{Testing the Parallel Task Scheduling algorithm.}
    \label{fig:runtime_benchmark}
    \Description{
        In this figure we see a plot with 'number of jobs' on the x-axis and 'time in ms' which the algorithm needs to run the instance on the y-axis.
        Mostly the plot shows that the algorithm runs in practical time but more explicit, we see that the algorithm needs about 250 ms to execute an instance with 40000 jobs, about 500 ms to execute an instance with 60000 jobs and about 750 ms to execute an instance with 80000 jobs.
    }
\end{figure*}

\begin{lemma}
    \label{lem:mcs94-1}
    If $N\geq 3$ then \cref{alg:mcs_absolute} has approximation ratio \(\frac{9}{4}\opt\).
\end{lemma}
\begin{proof}
    We consider the three cases.

    \case{1} \(N=3i\).
    In this case, we refer to the proof in~\cite{JansenR19}.

    \case{2} \(N=3i+1\).
    Here, the ratio of MCS1 is:
    \begin{align*}
             & \frac{12i+4}{6i+3}\opt+\frac{3}{6i+3}\pmax                           \\
        \leq & \frac{12+4}{6+3}\opt+\frac{3}{6+3}\opt     & i\geq 1, \pmax\leq \opt \\
        =    & \frac{19}{9}\opt<\frac{9}{4}\opt
    \end{align*}

    \case{3} \(N=3i+2\).
    Here, we reduce the makespan of the PTS schedule by at least \(\pmax\) by scheduling jobs which were completed after $h(\sigma)-\pmax$ on a separate cluster.
    Note that this cluster has a height of at most $2\pmax\leq 2\opt$.
    The remaining PTS schedule has then a makespan of at most \(\frac{4}{3}\opt_{\text{PTS}}\).
    Similar to the analysis in the proof of \cref{thm:mcs} we get:
    \begin{align*}
             & \frac{12i+8}{6i+3}\opt                                \\
        \leq & \frac{12+8}{6+3}\opt             & i\geq 1            \\
        =    & \frac{20}{9}\opt<\frac{9}{4}\opt &         & \qedhere
    \end{align*}
\end{proof}

The now remaining cases are \(N=1\) and \(N=2\).
Since \(N=1\) is PTS, we can use a \(2\)-approximation \cite{DBLP:journals/siamcomp/GareyG75,DBLP:conf/soda/LudwigT94,DBLP:conf/spaa/TurekWY92} to solve it.
Thus, we now focus on \(N=2\).

\begin{lemma}
    \label{lem:mcs94-2}
    There is an algorithm for MCS with approximation ratio \(\frac{9}{4}\opt\) for \(N=2\).
\end{lemma}
\begin{proof}
    Denote by \(A \coloneqq\sum_{j\in \calJ}p_jq_j\) the total area.
    Note that there are at most \(7\) jobs which have more than \(\frac{A}{8}\) area
    and there are at most \(2^{7-1}=64\) possibilities to assign those jobs to one of the two clusters.
    We perform the following steps for all these \(64=O(1)\) possible partitions \(P_1\sqcup P_2\)
    and can therefore assume that we found an optimal partition \(P_1^*\sqcup P_2^*\) of these jobs.
    Afterward we iteratively place the next job on the cluster with the least assigned area and execute \LS on these two partitions.
    W.l.o.g. assume that \(P_2^ *\) has more area assigned than \(P_1^*\).

    \case{1} \(\sum_{j\in P_2^*}p_jq_j-\sum_{j\in P_1^*}p_jq_j\leq \frac{1}{8}A\).
    Note that the partition with more area can have at most \(\frac{A}{8}\) area more than the other
    and therefore can have at most \(\frac{A}{2}+\frac{}{16}=\frac{9}{16}A\) area.
    The List-Scheduler schedules the jobs with height \(\leq 2\pmax\leq 2\opt\)
    or with height \(\leq2\cdot \frac{9}{16}A=\frac{9}{8}A\) \cite{DBLP:journals/siamcomp/GareyG75}.
    Using \(\frac{A}{2}\leq \opt\), which holds because \(\frac{A}{2}\) is the height of a perfect schedule
    where both clusters have the same height and at no time is a machine idle,
    we get \(\frac{9}{8}A\leq \frac{9}{4}\opt\).
    Furthermore, the height of the other cluster is also bounded by either \(2\pmax\) or \(A\leq 2\opt\).

    \case{2} \(\sum_{j\in P_2^*}p_jq_j-\sum_{j\in P_1^*}p_jq_j> \frac{1}{8}A\).
    In this case, the partition which had more area before the assignments of jobs
    with area \(\leq \frac{A}{8}\) has now still more area and no job with area \(\leq \frac{A}{8}\) assigned.
    Since we assumed that the partitions were optimal for the jobs with area \(>\frac{A}{8}\),
    this partition contains only a subset of jobs of an optimal partition \(P_2^{\opt}\).
    Therefore, the jobs are scheduled with height
    \(\leq 2\opt_{\text{PTS}}(P_2^*)\overset{P_2^*\subseteq P_2^{\opt}}{\leq} 2\opt_{\text{PTS}}(P_2^{\opt})= 2\opt\).
    Note that the other partition has now less than \(\frac{A}{2}\) area and is scheduled with height \(\leq 2\opt\) with the same argument as above.
\end{proof}
\cref{lem:mcs94-1} and \cref{lem:mcs94-2} directly imply our last result.
\mcs*

\section{Empirical Results}
In this section, we evaluate the computational performance of the PTS algorithm introduced in \cref{sec:pts}.
The experiments were conducted on a system equipped with an AMD Ryzen 9 3900X processor and 32 GB of RAM.
For our test instances, both the processing times and the machine requirements of the jobs were drawn from a uniform random distribution.
The results of this evaluation are presented in \cref{fig:runtime_benchmark}, illustrating the scalability of our approach.
The source code for our implementation is publicly available at \url{https://github.com/bennet-edler/PTS-Paper/}.

\section{Open Questions}
Since there does not exist an algorithm for \PTS with ratio smaller than $\frac{3}{2}$ unless $P=NP$,
eliminating the additive term is highly unlikely.
We note that there exists a pseudo-polynomial algorithm for the contiguous variant of this problem (strip packing)
with ratio $\frac{5}{4}+ \varepsilon$~\cite{JansenR19Pseudo}.
This naturally raises the question of whether our result can be further refined to obtain a linear-time algorithm
with an approximation ratio of $\frac{5}{4}\opt+\pmax$.

Additionally, since we heavily exploit the non-contiguousness of the problem,
a natural question to ask is whether our algorithm can be adapted to solve the contiguous version of the \PTS problem (strip packing).

\begin{acks}
    We thank the anonymous reviewers for their insightful feedback and helpful comments.

    Funded by the \grantsponsor{dfg}{Deutsche Forschungsgemeinschaft (DFG, German Research Foundation)}{https://gepris.dfg.de/gepris/projekt/453769249} – Project Number \grantnum[https://gepris.dfg.de/gepris/projekt/453769249]{dfg}{453769249}
\end{acks}

\bibliographystyle{ACM-Reference-Format}
\bibliography{src}

\end{document}